\pdfoutput=1 

\documentclass[final]{IEEEtran}
\IEEEoverridecommandlockouts
\usepackage{cite}
\usepackage{amsmath,amssymb,amsfonts,amsthm}
\usepackage{algorithmic}
\usepackage{graphicx}
\usepackage{textcomp}
\usepackage{xcolor}
\usepackage{color}
\usepackage{algorithmic}
\usepackage{url}
\usepackage{lipsum}
\usepackage{multirow}
\usepackage[normalem]{ulem}
\usepackage{afterpage}
\usepackage{bbm}
\usepackage{dsfont}
\usepackage[ruled,lined,boxed]{algorithm2e}
\usepackage{mathtools, nccmath}
\usepackage{lipsum}
\usepackage{bm}
\usepackage{threeparttable}
\usepackage[colorlinks=true,citecolor=black,linkcolor=black]{hyperref}
\usepackage{threeparttable}  
\usepackage{makecell}        
\usepackage{comment}
\usepackage{float}
\usepackage{dblfloatfix}
\usepackage{placeins}
\usepackage{booktabs}   
\usepackage{array}
\ifCLASSOPTIONcompsoc
\usepackage[caption=false, font=normalsize, labelfont=sf, textfont=sf]{subfig}
\else
\usepackage[caption=false, font=footnotesize]{subfig}
\fi

\DeclareMathAlphabet{\mathbcal}{OMS}{cmsy}{b}{n}
\def\BibTeX{{\rm B\kern-.05em{\sc i\kern-.025em b}\kern-.08em
		T\kern-.1667em\lower.7ex\hbox{E}\kern-.125emX}}

\newtheorem{lemma}{Lemma}
\newtheorem{proposition}{Proposition}
\newtheorem{remark}{Remark}

\newtheoremstyle{iremark}
{\topsep}   
{\topsep}  
{\upshape} 
{0pt}       
{\itshape} 
{.}         
{5pt plus 1pt minus 1pt} 
{\thmname{#1}\thmnumber{ \itshape#2}\thmnote{ (#3)}} 
\theoremstyle{iremark}

\def\blue{}

\mathchardef\mhyphen="2D
\makeatletter 
\let\myorg@bibitem\bibitem
\def\bibitem#1#2\par{%
	\@ifundefined{bibitem@#1}{%
		\myorg@bibitem{#1}#2\par
	}{%
		\begingroup
		\color{\csname bibitem@#1\endcsname}%
		\myorg@bibitem{#1}#2\par
		\endgroup
	}%
}

\makeatother

\expandafter\def\csname bibitem@8426033\endcsname{black}          
\expandafter\def\csname bibitem@pan2020multicell\endcsname{black} 
\expandafter\def\csname bibitem@li2023ergodic\endcsname{black}    
\expandafter\def\csname bibitem@abuzgaia2026fas\endcsname{black}  
\expandafter\def\csname bibitem@yu2004iterative\endcsname{black}  
\expandafter\def\csname bibitem@wang2026submodular\endcsname{black} 

\begin{document}

\title{Efficient Discrete Position Design for Movable Antenna Systems: Low Complexity and Robustness}
\author{
Haonan~Wang,~\IEEEmembership{Graduate~Student~Member,~IEEE},~Xianghao~Yu,~\IEEEmembership{Senior~Member,~IEEE}, \\Rui Wang,~Ang~Li,~\IEEEmembership{Senior~Member,~IEEE},~and~Ying-Jun~Angela~Zhang,~\IEEEmembership{Fellow,~IEEE}
\thanks{Haonan Wang and Xianghao Yu are with the Department of Electrical Engineering, City University of Hong Kong, Hong Kong (e-mail:
	haonwang2-c@my.cityu.edu.hk; alex.yu@cityu.edu.hk).}
\thanks{Rui Wang is with Microsoft Research Asia, Beijing 100080, China (e-mail: wrui0920@gmail.com).}
\thanks{Ang Li is with Shaanxi Key Laboratory of Deep Space Exploration and Intelligent Information Technology, the School of Information and Communications Engineering, Faculty of Electronic and Information Engineering, Xi’an Jiaotong University, Xi’an, Shaanxi 710049, China, and is also with the National Mobile Communications Research Laboratory, Southeast University, Nanjing, China (e-mail: ang.li.2020@xjtu.edu.cn).}
\thanks{Ying-Jun~Angela~Zhang is with the Department of Information Engineering, The Chinese University of Hong Kong, Hong Kong, China (e-mail: yjzhang@ie.cuhk.edu.hk).}
\thanks{Xianghao Yu is the corresponding author.}
\thanks{{\blue This paper was presented in part at the IEEE International Conference on Communications (ICC), Glasgow, UK, May 2026 \cite{wang2026submodular}.}}
}
\maketitle

\begin{abstract}
Building on advances in reconfigurable antenna techniques, movable antennas (MAs) can dynamically reshape antenna arrays and introduce additional spatial degrees of freedom (DoFs), thereby further improving communication performance. Despite these benefits, existing MA design algorithms often entail prohibitively high computational complexity from discrete positioning selection, which prevents practical implementations of MAs. In this paper, we investigate efficient solutions for the mutual information (MI) maximization problem of a multi-user multiple-input multiple-output (MU-MIMO) uplink communication system aided by discrete MAs. To this end, we first formulate the discrete MA positioning problem with the assumption of perfect channel state information (CSI). Then, we prove that the design problem falls into the category of monotone submodular maximization subject to a 2-system constraint. Accordingly, we propose a low-complexity distance-constrained submodular position search algorithm, which is theoretically shown to achieve at least  $\textcolor{black}{1/3}$ of the optimum. Furthermore, we extend our approach to scenarios with imperfect CSI, and show that the proposed submodular optimization-based design remains robust against channel estimation errors. Numerical results demonstrate that the proposed scheme can achieve at least 90\% of the optimal solution's MI gain under both perfect and imperfect CSI assumptions. Remarkably, the algorithm achieves orders-of-magnitude complexity reduction (e.g., $34.4\times$ faster than the branch-and-bound approach) while maintaining significant MI gains.
\end{abstract}

\begin{IEEEkeywords}
Fluid antenna, greedy search, $k$-system, MIMO, movable antenna, mutual information, position optimization, robust design, submodular.
\end{IEEEkeywords}

\IEEEpeerreviewmaketitle

\section{Introduction}
In the past few decades, antenna technology has profoundly catalyzed the development of wireless communication systems \cite{6798744}, paving the way for advanced multi-antenna techniques such as multiple-input multiple-output (MIMO) and its large-scale evolution, massive MIMO \cite{7400949}. The recent push towards even larger apertures and higher frequencies in technologies like extremely large-scale MIMO (XL-MIMO) is further driving the paradigm shift to near-field communications \cite{10845870}, where spherical wavefront propagation unlocks new spatial degrees of freedom in the range domain. While these innovations have significantly enhanced network capacity and reliability, conventional antenna deployments typically remain fixed in both position and geometry, limiting their adaptability to dynamic propagation environments. Recent advances in mechanically reconfigurable and liquid-based antennas, however, have introduced a new paradigm known as position reconfigurable antennas (PRAs) \cite{wang2025electromagnetically}. By dynamically adjusting antenna arrays, PRAs can actively {revamp} the wireless channel according to different communication requirements, endowing it with {desirable} characteristics and thereby further boosting communication performance in scenarios where fixed deployments fall short {\cite{10906511}}. 

From the perspective of practical implementation, PRAs can be realized in three principal forms: fluid antennas (FAs), metasurface-based antennas, and movable antennas (MAs). Specifically, FAs exploit shape-adaptive liquid radiators \cite{9982508}, while liquid metal designs have already verified the basic feasibility of PRAs \cite{shen2024design}. Meanwhile, metasurface-based solutions seek to emulate positional variation by selectively exciting pixels on coding metasurfaces \cite{wang2024multichannel,wang2023manipulations}, employing fast-switching positive-intrinsic-negative (PIN) diodes \cite{7086418}, or integrating substrate-integrated waveguides (SIWs) \cite{liu2025water}. Furthermore, MAs reposition the radiating element via direct mechanical actuation, as exemplified by the mechanically movable prototype reported in \cite{shao20246dma}. 
To strike a balance between performance gains and practical feasibility, and to further elucidate the physical mechanisms by which PRAs influence the wireless channel, MAs have garnered increasing attention in the field recently.

\subsection{\textcolor{black}{Related Works}}
MAs have been introduced into communication systems to further enhance network performance, and they can be broadly categorized into continuous and discrete {MAs}. A continuously adjustable MA represents the most general and ideal type of MA, assuming that the MA{s}’ position{s} can be selected from a finite continuous domain. In recent years, extensive research has been conducted on antenna position optimization for continuous MA systems in various classical communication scenarios, including MIMO (with single or multiple users) \cite{10243545}, near-field communications \cite{liu2025nearfieldcommunicationmassivemovable}, integrated sensing and communications (ISAC) \cite{10707252}, etc. From an optimization perspective, the continuous MA problem is highly non-convex and nonlinear. Consequently, the mainstream methods for position design primarily consist of gradient-based ascent/descent algorithms \cite{10243545}, particle swarm optimization (PSO) \cite{10620306}, and artificial intelligence-based (AI-based) approaches \cite{10534854}. Beyond these classical scenarios, continuously adjustable MAs have also been applied in emerging communication contexts such as space-air-ground integrated networks and over-the-air computation (AirComp). Specifically, an MA array was exploited in \cite{10806489} to minimize the average signal leakage power directed to interference areas, while maintaining a minimum beamforming gain over coverage areas, thereby addressing the undesirable sidelobes and interference leakage associated with traditional fixed antenna arrays. In addition, authors of \cite{10474119} introduced an MA array into AirComp systems to flexibly reshape the wireless channels between Internet of Things (IoT) devices and the sink node, further reducing the computation mean square error (CMSE) compared to conventional fixed-antenna-based AirComp systems.

While continuous MAs typically rely on mechanical approaches, the associated movement speed {and precision} constraints often preclude channel-level reconfiguration, thereby limiting their practical viability \cite{10753482}. In contrast, discrete MAs leverage fast-switching high-resolution array, enabling more rapid reconfiguration and making them more appealing for real-world applications. Moreover, as the number of spatial sampling points increases, discrete MAs can approximate {ideal} continuous MAs, which explains their growing adoption in recent research. Specifically, authors of \cite{10437926} investigated a downlink multi-user multiple-input single-output (MISO) broadcast system by formulating a transmit power minimization problem that jointly designs beamforming and discrete MA placement, deriving a globally optimal solution via a branch-and-bound (BnB) approach. Meanwhile, a single-user downlink MISO system was considered in \cite{10508218} and a graph optimization-based discrete MA design was further proposed to maximize the receive signal-to-noise ratio (SNR). Beyond these conventional scenarios, several studies have explored discrete MAs for diverse applications, such as physical-layer security and intelligent reflecting surface (IRS)-assisted communications. For instance, in \cite{10901621}, the authors discretized the transmit rail, formulated the MA-placement task as a path-selection problem on a multipartite graph, and solved it with an enumeration-with-pruning algorithm with a sequential update heuristic. In addition, the authors of \cite{wei2024movable} addressed a downlink IRS-assisted MISO scenario, jointly optimizing the MA positions, base station (BS) beamformer, and IRS phase shifts via a graph-based alternating optimization (AO) algorithm.

\subsection{Motivations}
In principle, prevailing algorithms for discrete MA placement continue to rely on exhaustive enumeration whose worst-case complexity scales exponentially with the array dimension. In particular, the BnB family epitomizes this limitation and rapidly becomes computationally prohibitive even for a moderate number of antenna elements \cite{10437926}, which severely deviates from the massive MIMO evolution in modern wireless communication systems. Furthermore, the majority of existing contributions formulated the placement task of discrete MAs as a generic mathematical program and subsequently call upon off-the-shelf solvers, e.g., BnB or graph-based method, without leveraging the combinatorial and geometric structure inherent to discrete MAs. While such approaches remain capable of producing feasible solutions, they overlook the distinctive characteristics of the problem, yielding results that are neither scalable nor revealing. {\blue In particular, their computational cost becomes prohibitive as the system parameters increase, and the resulting position configurations provide limited insight into how antenna positioning reshapes the channel.} Motivated by these observations, we seek to exploit the discrete structure of the MA configuration itself and develop a bespoke algorithmic framework. Our approach aims to deliver near-optimal {\blue mutual information (MI)} at a fraction of the computational cost, thereby enriching theoretical insight and facilitating practical deployment.
\subsection{\textcolor{black}{Contributions}}
In this paper, we investigate the design of {discrete} MA positions to maximize the MI in an uplink multi-user MIMO (MU-MIMO) system. To comprehensively capture the impact of MAs on wireless channels, we study the MA design problem under both perfect {and imperfect} channel state information (CSI) assumptions. Through rigorous analysis of the mathematical properties {of the formulated problem}, we propose efficient design algorithms with low computational complexity. The main contributions are summarized as follows:

\begin{itemize}
	\item Under the assumption of perfect CSI, we first formulate a {discrete} MA position design problem, to maximize the uplink MI in MU-MIMO systems. We then prove that the objective function of the formulated problem falls into the category of monotone submodular functions subject to a \textcolor{black}{2-system} constraint.
	
	\item To address the discrete MA positioning problem, we develop a low-complexity \textcolor{black}{distance-constrained submodular position search (DCSPS)} algorithm that enforces minimum antenna spacing constraints while maintaining polynomial-time complexity. Unlike existing approaches in \cite{10243545,10535440 }, we provide a theoretical performance guarantee for the proposed algorithm: It achieves \textcolor{black}{at least} a $\textcolor{black}{1/3}$-approximation ratio relative to the optimum.
	
	\item We further consider the robust MA position design problem under \textcolor{black}{the assumption of imperfect CSI}. Specifically, by adopting an additive statistical channel estimation error model, we formulate a robust design problem aiming to maximize the expected MI, and further extend the submodularity of MA position design to the imperfect CSI scenario. \textcolor{black}{Based on this, we propose a robust DCSPS (R-DCSPS) scheme that comes with a performance guarantee for the surrogate objective.}
\end{itemize}

\textcolor{black}{Numerical results demonstrate that, under both perfect and imperfect CSI, the proposed scheme achieves over 90\% of the MI achieved by the optimal solution over traditional fixed-antenna systems, while reducing computational complexity by 34.4 times compared to the BnB method. The proposed \textcolor{black}{DCSPS} design ensures an effective balance between MI and complexity, and demonstrates enhanced robustness in the presence of channel estimation errors.}

\textit{Notations}: Let \(a\), \(\mathbf{a}\), and \(\mathbf{A}\) denote scalar, vector, and matrix, respectively. {The imaginary unit is denoted by $\jmath \triangleq \sqrt{-1}$.} The operators \((\cdot)^{\ast}\), \((\cdot)^{T}\), \((\cdot)^{H}\), and \(\operatorname{Tr}(\cdot)\) denote the conjugate, transpose, conjugate transpose, and trace of a matrix, respectively. $\operatorname{det}\left(\cdot\right)$ denotes the determinant of a square matrix. The operator \(\operatorname{vec}(\cdot)\) denotes the vectorization of a matrix, and \(\operatorname{diag}(\cdot)\) generates a diagonal matrix whose diagonal elements are extracted from a vector. $\mathrm{blkdiag}(\mathbf{A}_1,\dots, \mathbf{A}_n)$ denotes the block diagonal matrix formed by placing matrices $\mathbf{A}_1, \dots, \mathbf{A}_n$ along the main diagonal. The Frobenius norm of a matrix is denoted by \(\|\cdot\|_{F}\). The defining equality is denoted by $\triangleq$. The notations \(\mathbb{R}^{M \times N}\), \(\mathbb{C}^{M \times N}\), and \(\mathbb{N}^{M \times N}\) refer to the sets of real, complex, and integer matrices of dimension \(M \times N\), respectively. The matrix \(\mathbf{I}_{M}\) denotes the \(M \times M\) identity matrix. The notation \(\mathbf{A} \succeq \mathbf{0}\) means that \(\mathbf{A}\) is positive semidefinite, whereas \(\mathbf{A} \geqslant \mathbf{0}\) indicates that \(\mathbf{A}\) is entry-wise non-negative. \( a \in \mathcal{A} \) denotes that \( a \) is an element of the set \( \mathcal{A} \), while \(\lvert \mathcal{A} \rvert\) denotes the cardinality of the set \(\mathcal{A}\). Finally, \(\mathcal{A} \cap \mathcal{B}\) and \(\mathcal{A} \cup \mathcal{B}\) represent the intersection and union of sets \(\mathcal{A}\) and \(\mathcal{B}\), respectively. $\emptyset$ denotes an empty set. \( \mathcal{A} \subseteq \mathcal{B} \) denotes that \( \mathcal{A} \) is a subset of \( \mathcal{B} \). $ \mathcal{V} \setminus \{v\}$ denotes the set difference, representing all elements in $\mathcal{V}$ except the element $v$. $\{ v \in \mathcal{V} : \text{condition} \}$ defines the set of all elements $v$ from $\mathcal{V}$ that satisfy the given condition. \( \mathbf{n} \sim \mathcal{CN}(\boldsymbol{\mu}, \mathbf{C}) \) indicates that the random vector \( \mathbf{n} \) follows a circularly symmetric complex Gaussian distribution with mean vector \( \boldsymbol{\mu} \) and covariance matrix \( \mathbf{C} \).

\section{System Model and Problem Formulation}

In this section, we present the mathematical model of a discretely adjustable movable antenna {MIMO} (MA-MIMO) channel in a multi-user uplink communication scenario. Subsequently, we formulate the MI maximization problem by designing the antenna index vector (AIV).

\subsection{MA-MIMO System Model}
\begin{figure}[t]
	\centering
	\includegraphics[width=3in]{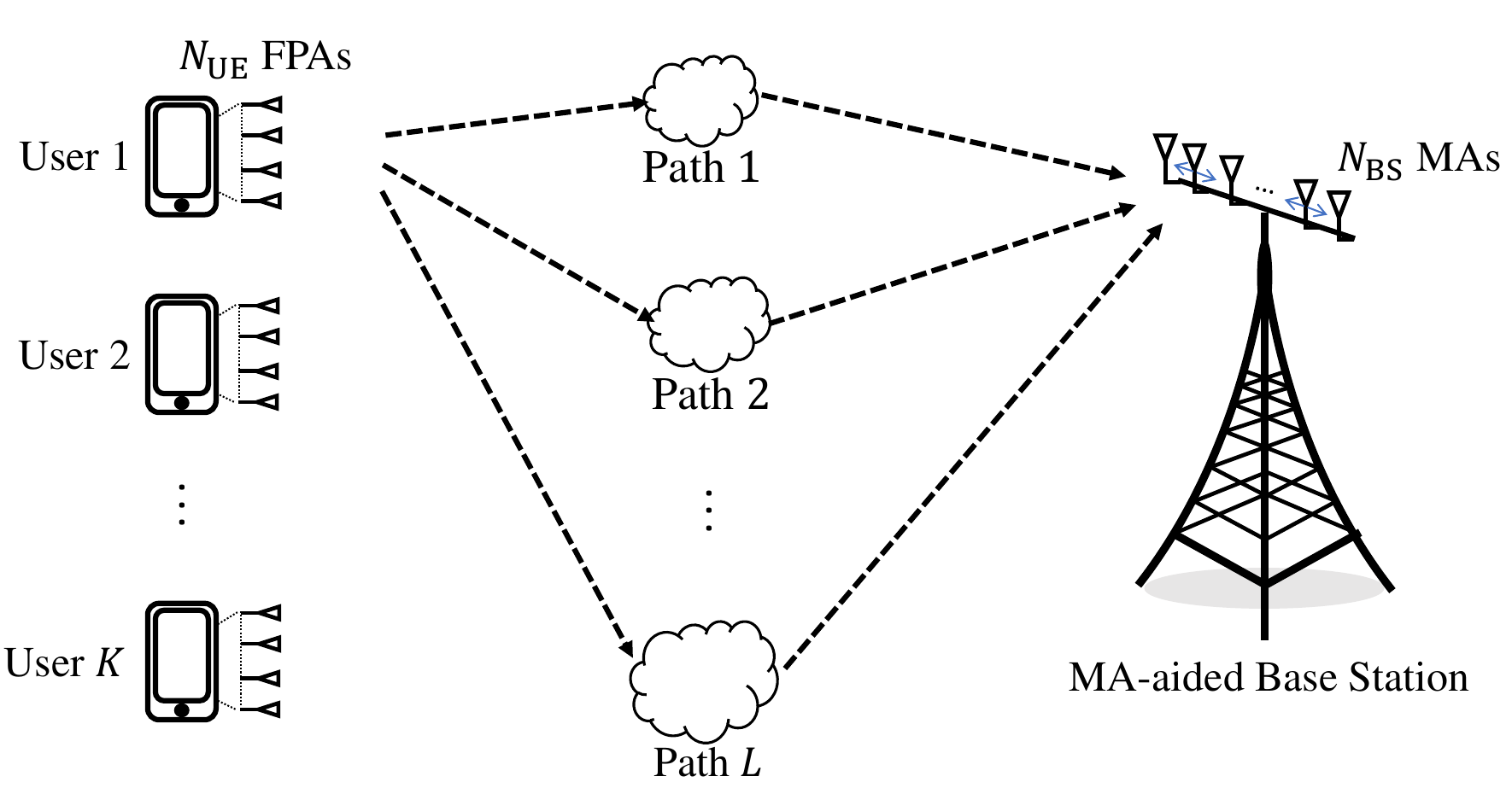}
	\caption{{{An u}plink MA-aided MU-MIMO communication system.}}
	\label{fig1}
\end{figure}

Consider a single-cell MU-MIMO uplink communication system, as shown in Fig.~\ref{fig1}, consisting of one BS and $K$ users. Specifically, each user deploys \(N_{\mathrm{U}}\) fixed-position antennas (FPAs) arranged with half-wavelength spacing on a uniform linear array (ULA), while the BS is equipped with \(N_{\mathrm{BS}}\) MAs{\footnote{{Although this study confines MA deployment to the BS, its conclusions still apply when MAs are also deployed at each user, as guaranteed by the AO principle \cite{10243545}.}}} constrained to move along a one-dimensional linear array\footnote{While this paper focuses on linear arrays for a neat presentation, the conclusions can be readily extended to planar arrays.}. The antenna position vector for MAs {is given by}~\cite{10243545}
\begin{equation}    
	\mathbf{r}=\left[r_{1},r_{2},\ldots,r_{N_{\mathrm{BS}}}\right]\in \mathbb{R}^{1\times N_{\mathrm{BS}}},
	\label{Equ.1}
\end{equation}
where \(r_{m}\) denotes the {position} of the \(m\)-th MA.

In this paper, we consider a discrete MA system for practical implementation, where each MA can occupy only one distinct position among $N_\mathrm{S}$ predefined uniform grids. Without loss of generality, we assume $N_\mathrm{S}\gg N_\mathrm{BS}$ to facilitate discrete MA deployment. An example of a discretely adjustable MA array and a comparison between the FPA are presented in Fig.~\ref{figma}. By denoting the AIV of MAs as
\begin{equation}    
	\mathbf{v}=\left[v_{1},v_{2},\ldots,v_{N_{\mathrm{BS}}}\right]
	\in \mathbb{N}^{1\times N_{\mathrm{BS}}},
\end{equation}
where {$v_m \in \{1,2, \cdots, N_{\mathrm{S}}\}$} is the integer-valued position {index} for the \(m\)-th MA, the position of each MA in \eqref{Equ.1} is further represented as
\begin{equation}    
	\mathbf{r}=\frac{{A}}{N_{\mathrm{S}}}\mathbf{v},
	\label{Equ.2}
\end{equation}
where $A$ denotes the aperture of the array. For conventional fixed ULAs, 
the antenna spacing is typically fixed as \(\lambda/2\), i.e., $v_m=m$, $N_\mathrm{S}= N_\mathrm{BS}$, and $A=(N_{\mathrm{BS}}-1)\lambda/2$, where $\lambda$ denotes the carrier wavelength.

The received field response matrix \(\mathbf{F}_k(\mathbf{v})\) at the BS corresponding to the \(k\)-th user is expressed as
\begin{equation}
	\mathbf{F}_k(\mathbf{v})
	= \left[ \mathbf{f}_k(v_{1}), \mathbf{f}_k(v_{2}), \ldots, \mathbf{f}_k(v_{N_{\mathrm{BS}}}) \right]
	\in \mathbb{C}^{L \times N_{\mathrm{BS}}},\label{Fk}
\end{equation}
where {$L$ is the number of scattering paths, and} the field response vector \(\mathbf{f}_k(v_{m})\) corresponding to the \(m\)-th MA is given by
\begin{equation}
	\mathbf{f}_k(v_{m})
	= \left[ e^{{\jmath}\frac{2\pi A}{N_{\mathrm{S}}\lambda}v_{m}\cos\phi_{1,k}},\ 
	\ldots,\ 
	e^{{\jmath}\frac{2\pi A}{N_{\mathrm{S}}\lambda}v_{m}\cos\phi_{L,k}} \right]^{T}
	\in \mathbb{C}^{L \times 1},
\end{equation}
with \(\phi_{l,k}\) denoting the azimuth angle of arrival (AoA) of the \(l\)-th scattering path for the \(k\)-th user.
\begin{figure}[t]
	\centering
	\includegraphics[width=3.2in]{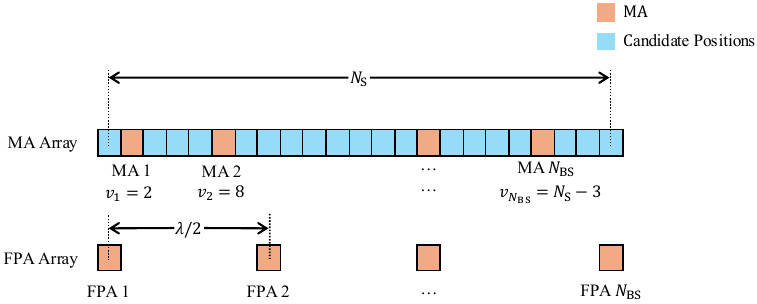}
	\caption{Discretely adjustable MA array with $N_\mathrm{BS}$ antenna elements and $N_\mathrm{S}$ possible positions.}
	\label{figma}
\end{figure}

Meanwhile, with the FPA at the users, the transmit field response matrix at the $k$-th user \textcolor{black}{is given by}
\begin{equation}
	\mathbf{G}_k
	=\left[\mathbf{g}_{1,k},\mathbf{g}_{2,k},\ldots,\mathbf{g}_{N_{\mathrm{U}},k}\right]
	\in \mathbb{C}^{L \times N_{\mathrm{U}}},\label{G}
\end{equation}
where
\begin{equation}
	\mathbf{g}_{n,k}
	=\left[e^{{\jmath} \pi (n-1) \cos \theta_{1,k}},
	\ldots,
	e^{{\jmath} \pi (n-1) \cos \theta_{L,k}}
	\right]^{T}
	\in \mathbb{C}^{L \times 1},\label{g}
\end{equation}
and \(\theta_{l,k}\) denotes the azimuth angle of departure (AoD) for the \(l\)-th scattering path corresponding to the $k$-th user. Therefore, the uplink MA-MIMO channel matrix for the $k$-th user is expressed as
\begin{equation}
	\mathbf{H}_k\left(\mathbf{v}\right)
	=\mathbf{F}_k^{H}(\mathbf{v})
	\mathbf{\Sigma}_k
	\mathbf{G}_k,
	\label{eq1}
\end{equation}
where \(\mathbf{\Sigma}_k=\mathrm{diag}\left(\alpha_{1,k},\alpha_{2,k},\ldots,\alpha_{L,k}\right)\), with the complex path gain $\alpha_{l,k}$ corresponding to the $l$-th scattering path from the $k$-th user to the BS.
Finally, the received signal at the BS is given by
\begin{equation}
	\mathbf{y}
	=\sum_{k=1}^{K}\mathbf{H}_k\left(\mathbf{v}\right)\mathbf{x}_k
	+\mathbf{n}{,}\label{9}
\end{equation}
where \(\mathbf{x}_k\) denotes the transmit signal from the $k$-th user. It is modeled as complex circularly symmetric with $\mathbb{E}\left\{\mathbf{x}_k^{H}\mathbf{x}_k\right\} \le P$, where \(P\) denotes the maximum transmit power. In addition, \(\mathbf{n}\sim \mathcal{CN}\left(\mathbf{0},\sigma^{2}\mathbf{I}_{N_{\mathrm{BS}}}\right)\) represents the additive white Gaussian noise (AWGN) at the BS, with $\sigma^2$ being the noise power.

\subsection{\texorpdfstring{MI}{MI} Maximization via AIV Design}
We now discuss the AIV design problem for MI maximization in the considered discretely adjustable MA-aided uplink MU-MIMO system. Specifically, the sum-rate of an uplink MU-MIMO system is upper bounded by the MI \cite{tse2005fundamentals}
\begin{equation}
	\sum_{k=1}^{K} R_k \leq \log_2 \det \left(\mathbf{I}_{N_{\mathrm{BS}}} 
	+\rho\sum_{k=1}^K
	\mathbf{H}_{k}(\mathbf{v})\mathbf{Q}_k\mathbf{H}_{k}^{H}(\mathbf{v})\right), \label{region}
\end{equation}
where $R_k$ denotes the achievable rate of the $k$-th user, $\rho=1/\sigma^2$ denotes the normalized SNR, and {$\mathbf{Q}_k=\mathbb{E}\left\{\mathbf{x}_k\mathbf{x}_k^H\right\}$} is the transmit covariance matrix of the $k$-th {user}. Here, $\mathbf{H}_k(\mathbf{v})$ denotes the obtained CSI associated with the $k$-th user, which may be perfectly known or subject to estimation errors. This theoretical upper bound can be achieved by adopting the minimum mean-square error receiver with successive interference cancellation (MMSE-SIC) \cite[Sec. 10.2]{tse2005fundamentals}.

In \eqref{region}, a rate region is described through the MI, which is the theoretical upper bound of the achievable sum-rate. Within this region, different transmit strategies yield different rate tuples. In essence, this region characterizes the “potential” of the communication channel. Therefore, we discuss the MI maximization via MA design, to expand the rate region as much as possible. This formulation quantifies the {\blue MI gain enabled by MA and reveals how antenna positioning reshapes the wireless channel \cite{8426033,pan2020multicell}.}

The MI maximization problem is formulated as
\begin{equation}
	\begin{aligned}
		\mathcal{P}_0:\max _{{\mathbf{v}}, \mathbf{Q}_k} & \quad \log _2 \operatorname{det}\left(\mathbf{I}_{N_{\mathrm{BS}}}+\rho\sum_{k=1}^{K} \mathbf{H}_k(\mathbf{v})\mathbf{Q}_k \mathbf{H}_k^{H}(\mathbf{v})\right) \\
		\text { s.t. } & \quad \text{C1: } 1 \leq v_{m} \leq N_{\mathrm{S}}, \\
		& \quad {\text{C2: } \mathbf{v} \in \mathbb{N}^{1 \times N_{\mathrm{BS}}},} \\
		& \quad \text{C3: } \left|v_{n}-v_{m}\right| \geq N_{\mathrm{D}}, \ {\forall}\ 1 \leq n,m \leq N_{\mathrm{BS}}, n \neq m, \\
		& \quad {\text{C4: }\operatorname{Tr}(\mathbf{Q}_k) \leq P}, \\
		& \quad \text{C5: } \mathbf{Q}_k \succeq \mathbf{0}.
	\end{aligned}
	\label{P0}
\end{equation}
	
In problem $\mathcal{P}_0$, C1 and C2 ensure that each MA is restricted to a discrete position among the predefined grids. C3 imposes a minimum spacing of {$N_{\mathrm{D}}$} grid units between MAs to maintain sufficient spatial separation and mitigate mutual interference. C4 limits the transmit power, while C5 specifies that the transmit covariance matrices $\mathbf{Q}_k$'s are positive semidefinite (PSD). Typically, the system has two primary design degrees of freedom (DoFs): the AIV \(\mathbf{v}\) and the transmit covariance matrices \(\mathbf{Q}_k\)'s. To focus on characterizing the fundamental limits of discrete MA-assisted wireless systems, in this paper we assume isotropic transmissions, i.e., \(\mathbf{Q}_k = \frac{P}{N_{\mathrm{U}}}\mathbf{I}_{N_{\mathrm{U}}}\), for the users. This setting allows us to highlight the MI brought by MAs. Accordingly, problem $\mathcal{P}_0$ is transformed as
\begin{equation}
	\begin{aligned}
		\mathcal{P}_1:\max _{{\mathbf{v}}} & \quad \log _2 \operatorname{det}\left(\mathbf{I}_{N_{\mathrm{BS}}}+\rho\sum_{k=1}^{K} \mathbf{H}_k(\mathbf{v}) {\mathbf{H}_k^{H}(\mathbf{v})}\right) \\
		\text { s.t. } & \quad \text{C1: } 1 \leq v_{m} \leq N_{\mathrm{S}}, \\
		& \quad {\text{C2: } \mathbf{v} \in \mathbb{N}^{1 \times N_{\mathrm{BS}}},} \\
		& \quad \text{C3: } \left|v_{n}-v_{m}\right| \geq N_{\mathrm{D}}, \ {\forall}\ 1 \leq n,m \leq N_{\mathrm{BS}}, n \neq m. 
	\end{aligned}
	\label{P1}
\end{equation}
\begin{remark}
	{\blue Since channel estimation is typically performed at the BS through uplink pilot training, instantaneous channel state information at the transmitter (CSIT) is generally unavailable to the distributed users. We therefore adopt the channel-independent isotropic covariance \(\mathbf{Q}_k=\frac{P}{N_{\mathrm{U}}}\mathbf{I}_{N_{\mathrm{U}}}\) in the considered uplink system.} For scenarios requiring joint optimization of AIV \(\mathbf{v}\) and transmit covariance matrices \(\mathbf{Q}_k\)'s, an AO approach can be employed \cite{10243545}. In such a framework, the MA positioning subproblem (optimizing \(\mathbf{v}\) with fixed \(\mathbf{Q}_k\)'s) can leverage the submodular optimization techniques developed in this paper, while the transmit covariance matrices \(\mathbf{Q}_k\)'s can be optimized in parallel by well-established beamforming methods such as singular value decomposition-based transceivers with water-filling power allocation.
	\label{remark:1}
\end{remark}

By denoting $\mathbf{H}(\mathbf{v})=\left[\mathbf{H}_{1}(\mathbf{v}),\mathbf{H}_{2}(\mathbf{v}),\ldots,\mathbf{H}_{K}(\mathbf{v})\right]$,
we obtain a compact version of AIV optimization problem, as shown below.
\begin{equation}
	\begin{aligned}
		\mathcal{P}_{2}: 
		\max_{\mathbf{v}} 
		& \quad  c(\mathbf{v})\triangleq\log_2 \det \left(\mathbf{I}_{N_{\mathrm{BS}}}
		+\rho
		\mathbf{H}(\mathbf{v})\mathbf{H}^{H}(\mathbf{v})\right)\\
		\text { s.t. } & \quad \text{C1: } 1 \leq v_{m} \leq N_{\mathrm{S}}, \\
		& \quad {\text{C2: } \mathbf{v} \in \mathbb{N}^{1 \times N_{\mathrm{BS}}},} \\
		& \quad \text{C3: } \left|v_{n}-v_{m}\right| \geq N_{\mathrm{D}}, \ {\forall} 1 \leq n,m \leq N_{\mathrm{BS}}, n \neq m.
	\end{aligned}
	\label{P5}
\end{equation}
\begin{remark}
	From problem~$\mathcal{P}_2$, we observe that the discretely adjustable MA design problem is an integer and combinatorial optimization problem, and solving it by exhaustive search is NP-hard. Even the widely-used BnB methods entail prohibitively high computational complexity (exponential in the MA array size) in practical applications. Therefore, a more computationally-efficient algorithm that attains a high MI is highly desirable. In this paper, we shall propose a submodular optimization framework for MI maximization by AIV design under different CSI assumptions.
	\label{remark:2}
\end{remark}
\section{{Preliminaries on Submodular Optimization}}
In this section, we first introduce the fundamental concepts of submodularity, including its definition and key properties. We then discuss the relationship between submodularity and concavity to provide a deeper understanding of the subject.
\subsection{Key Definitions and Main Properties}
\vspace{1ex}
\noindent
\textit{Definition 1 (Submodularity) }\cite{petersen1987stabilization,kelmans1983multiplicative}:
Consider a ground set of $n$ objects $\mathcal{V} \triangleq \{v_1,\dots,v_n\}$ and a set-function $f: 2^{\mathcal{V}} \to \mathbb{R}$ that assigns a real value to each $\mathcal{S} \subseteq \mathcal{V}$. The function $f$ is submodular if, for all $\mathcal{A}, \mathcal{B} \subseteq \mathcal{V}$,
\begin{equation}
	f(\mathcal{A}) + f(\mathcal{B})
	\;\;\ge\;\;
	f(\mathcal{A}\cup\mathcal{B}) \;+\;
	f(\mathcal{A}\cap\mathcal{B}).
	\label{Equation.1}
\end{equation}

To understand why inequality \eqref{Equation.1} holds and how it characterizes submodularity, we reclaim this concept using an equivalent definition of submodular functions based on incremental gains, thereby making the principle of diminishing returns more explicit. Specifically, we define the \emph{incremental gain} of adding an element $v \in \mathcal{V}\setminus \mathcal{A}$ to a subset $\mathcal{A} \subseteq \mathcal{V}$ as
\begin{equation}
	\Delta_f(v \mid \mathcal{A})\triangleq f\left(\mathcal{A}\cup\{v\}\right)
	\;-\;
	f(\mathcal{A}).
\end{equation}
Based on that, \eqref{Equation.1} can be equivalently reformulated as follows:

If $f$ is submodular, then for all
$\mathcal{A} \subseteq \mathcal{B} \subseteq \mathcal{V}$ and $v \in \mathcal{V}$,
\begin{equation}
	\Delta_f(v \mid \mathcal{A})
	\;\;\geq\;\;
	\Delta_f(v \mid \mathcal{B}),
	\label{Equation.2}
\end{equation}
which illustrates that the incremental gain for an element declines as we grow the underlying set. If \eqref{Equation.1}--\eqref{Equation.2} hold with equality, $f$ is said to be \emph{modular}, meaning each element contributes independently of others.

In other words, submodularity can be seen as a formal manifestation of the diminishing returns principle: Once a certain set of elements has been selected, the additional value gained by including a new element is smaller (or at most equal) than when adding that same element to a smaller set. Submodularity therefore encodes a natural “penalty” for repeatedly adding new elements to an already large set, preventing the overall value from inflating too rapidly and ensuring that the impact of each successive element remains bounded. In essence, this property is analogous to concavity in continuous-domain optimization problems, which we discuss in the next subsection.


\vspace{1ex}
\noindent
\textit{Definition 2 (Monotonicity)}:
Submodularity often appears jointly with \emph{monotonicity}, which states that $f(\mathcal{A}) \le f(\mathcal{B})$ whenever $\mathcal{A} \subseteq \mathcal{B} \subseteq \mathcal{V}$. For a submodular function $f$, this is equivalent to
\begin{equation}
	\Delta_f(v \mid \mathcal{V} \setminus \{v\})
	\;\;\ge\;\;
	0,
	\quad
	\forall v \in \mathcal{V}.
\end{equation}

Monotonicity guarantees non-negative incremental gains from each additional element. In the scenario we consider, it ensures that choosing more antennas never decreases the MI, thus facilitating incremental selection strategies.

\vspace{1ex}
\noindent
\textit{Definition 3 ($k$-System)} \cite{feldman2011improved}:
A set system $(\mathcal{V}, \mathcal{I})$, where $\mathcal{I} \subseteq 2^{\mathcal{V}}$ is a collection of independent sets, is called a $k$-system if the following conditions are met:
\begin{enumerate}
	\item $\emptyset \in \mathcal{I}$ (non-emptiness property).
	\item If $\mathcal{A} \in \mathcal{I}$ and $\mathcal{B} \subseteq \mathcal{A}$, then $\mathcal{B} \in \mathcal{I}$ (heredity property).
	\item For any subset $\mathcal{S} \subseteq \mathcal{V}$, the ratio of the sizes of the largest and smallest maximal independent sets within $\mathcal{A}$ is bounded by $k$, i.e.,
	\begin{equation}
		\frac{\xi_{\mathrm{max}}(\mathcal{S})}{\xi_{\mathrm{min}}(\mathcal{S})} \leq k,
	\end{equation}
	where $\xi_{\mathrm{max}}(\mathcal{S})$ and $\xi_{\mathrm{min}}(\mathcal{S})$ denote the sizes of the largest and smallest maximal independent set contained in $\mathcal{S}$, respectively. An independent set $\mathcal{T} \in \mathcal{I}$ is maximal within $\mathcal{S} \subseteq \mathcal{V}$ if $\mathcal{T} \subseteq \mathcal{S}$ and no element can be added to it without violating independence, i.e., $\forall x \in \mathcal{S} \setminus \mathcal{T},  \mathcal{T} \cup \{x\} \notin \mathcal{I}$.
\end{enumerate}

The above definition provides a mathematically rigorous description of submodularity. To gain a deeper understanding of the essence of submodularity, we will discuss the relationship between submodularity and concavity in the following part.

\subsection{{Submodularity and Concavity}}

Submodularity is often viewed as a discrete analog of concavity, primarily through the concept of diminishing incremental gains \cite{fujishige1978polymatroidal,kelmans1983multiplicative}. In the continuous domain, a concave function exhibits diminishing returns via a non-increasing derivative. While, in the discrete domain, submodularity imposes that adding an element \(v\) to a larger set yields the same or smaller incremental benefit than adding it to a smaller set. This discrete “flattening slope” directly parallels the continuous case.

Fig.~\ref{structure} provides a visual comparison. The left panel illustrates a concave function whose slope decreases as the input grows, while the right panel shows a hypercube of all subsets of a ground set. Edges in the hypercube connect subsets in a containment relationship, highlighting that moving from a smaller subset to a larger one reduces the “discrete derivative”, just as a concave function’s slope flattens.

\begin{figure}[ht]
	\centering
	\includegraphics[width=3in]{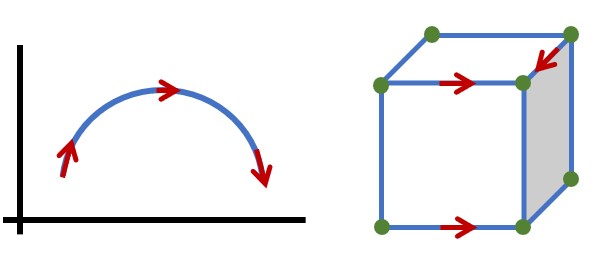}
	\caption{A unified view of diminishing incremental gains for concave and submodular functions.}
	\label{structure}
\end{figure}

Submodularity’s diminishing-return (DR) property is formally captured in \eqref{Equation.2}, ensuring that incremental gains decay as the chosen set expands. In continuous optimization, concavity can often be tackled via gradient-based or specialized methods. Analogously, for submodular set-function maximization, a simple greedy algorithm leverages the DR structure by repeatedly selecting the element with the highest incremental contribution, akin to taking the “steepest slope” in concavity.

These parallels yield powerful approximation results in the discrete setting. Although maximizing a monotone submodular function \(f\) under a \textcolor{black}{$k$-system} constraint is NP-hard, the aforementioned greedy approach guarantees an approximation ratio \textcolor{black}{of $(1+k)^{-1}$ relative to the optimal value \cite{feldman2011improved}}. This result stems from precisely the same DR principle that makes concavity amenable to efficient methods in the continuous domain.

\section{Submodular Optimization with Perfect CSI}
In this section, we investigate the AIV optimization problem under the assumption of perfect CSI. Specifically, we first reformulate problem $\mathcal{P}_2$ as a discrete selection process from a predefined codebook. Based on that, we show that the resulting problem has a monotone submodular objective function, and further propose a low-complexity DCSPS algorithm.
\subsection{Problem Reformulation}
While the optimization in {$\mathcal{P}_2$} provides a general MI maximization, we highlight the selection nature of the AIV design by specifying discrete sets for MAs. Let
\begin{equation}
	\mathcal{W} \triangleq \{1,2,\dots, N_{\mathrm{S}}\}
\end{equation}
which contains all the possible position indices for MAs. In principle, for discrete MA design, we select $N_{\mathrm{BS}}$ elements from $\mathcal{W}$ (with minimum spacing $N_{\mathrm{D}}$) for AIV, and these selected elements constitute a set $\mathcal{S}$. 
Specifically, we have
\begin{equation}
	\mathcal{S} \subseteq \mathcal{W}, \quad |\mathcal{S}| = {N_{\mathrm{BS}}}, \quad 
	|i-j| \ge N_{\mathrm{D}}, \quad \forall i,j \in \mathcal{S}.
\end{equation}

Once set $\mathcal{S}$ is chosen, the resulting AIV adopts corresponding entries from $\mathcal{S}$. We define the predefined codebook channel as
\begin{equation}
	\bar{\mathbf{H}}=\bar{\mathbf{F}}^{H}\mathbf{\Sigma}{\mathbf{G}} \in \mathbb{C}^{N_{\mathrm{S}} \times {KN_{\mathrm{U}}}}, \label{eq:codebook channel}
\end{equation}
where 
\begin{equation}
	\mathbf{\Sigma}=\operatorname{blkdiag}\left(\mathbf{\Sigma}_1,\mathbf{\Sigma}_2,\ldots,\mathbf{\Sigma}_K\right) \in \mathbb{C}^{KL \times KL},\label{Eq:Sigma}
\end{equation}
\begin{equation}
	\mathbf{G}=\operatorname{blkdiag}\left(\mathbf{G}_1,\mathbf{G}_2,\ldots,\mathbf{G}_K\right) \in \mathbb{C}^{KL \times KN_{\mathrm{U}}},\label{Eq:AllG}
\end{equation}
and
\begin{equation}
	\bar{\mathbf{F}}=\left[\bar{\mathbf{F}}_1^H,\bar{\mathbf{F}}_2^H,\ldots,\bar{\mathbf{F}}_K^H\right]^H \in \mathbb{C}^{KL \times N_{\mathrm{S}}},
	\label{eq:codebook F}
\end{equation}
with
\begin{equation}
	\bar{\mathbf{F}}_k=\left[\mathbf{f}_k(1),\mathbf{f}_k(2),\ldots,\mathbf{f}_k(N_{\mathrm{S}})\right] \in \mathbb{C}^{L \times N_{\mathrm{S}}}.
	\label{eq:codebook Fk}
\end{equation} 
Then, the selected MA-MIMO channel is written as 
\begin{equation}
	\mathbf{H}_{\mathcal{S}}
	= 
	\bar{\mathbf{F}}_{\mathcal{S}}^{H} 
	\mathbf{\Sigma}
	{\mathbf{G}}, 
	\label{eq:selected H}
\end{equation}
where $\bar{\mathbf{F}}_{\mathcal{S}}$ collects column vectors from $\bar{\mathbf{F}}$ whose indices stored in set $\mathcal{S}$.

The AIV optimization problem {$\mathcal{P}_2$ is} thus reformulated as
\begin{equation}
	\begin{aligned}
		\mathcal{P}_{3}: 
		\max_{\mathcal{S}} 
		& \quad {\blue c(\mathcal{S})}=\log_2 \det \left(\mathbf{I}_{\left|\mathcal{S}\right|}
		+\rho
		\mathbf{H}_{\mathcal{S}}\mathbf{H}_{\mathcal{S}}^{H}\right)\\
		\text{s.t.} 
		& \quad \text{C1:}\ \mathcal{S} \subseteq \mathcal{W}, \\
		& \quad {\text{C2:}\ |\mathcal{S}| = N_{\mathrm{BS}},}  \\
		&\quad \text{C3:}\ |i-j| \ge N_{\mathrm{D}}, \;\forall i,j \in \mathcal{S}.
	\end{aligned}
	\label{P2}
\end{equation}

The transformation from {problem} $\mathcal{P}_2$ to $\mathcal{P}_3$ offers a set-based perspective, capturing the dynamic evolution of the feasible region and enabling the derivation of global performance guarantees from local incremental gains. 
Generally, {problem} $\mathcal{P}_3$ is an integer programming problem and is NP-hard in the worst case. However, in the following we show that the objective function of $\mathcal{P}_3$ exhibits a monotone submodular property, which enables a DCSPS solution with a guaranteed performance bound.

\subsection{Submodularity}
The optimization problem $\mathcal{P}_3$ in \eqref{P2} is solved by a sequential selection process, where at each step, one MA position is added to the set $\mathcal{S}$. Thus, at the $m$-th step of this process, a partial set of $m-1$ positions has already been selected, and the $m$-th position is chosen from the remaining feasible candidates in $\mathcal{W}$ that satisfy the minimum spacing constraint.
Given that submodularity captures the evolution of the objective function’s incremental gain during the selection process, {\blue we first quantify the incremental gain for an arbitrarily selected set in the following lemma.}
\begin{lemma}
	{\blue For any $\mathcal{S}\subseteq\mathcal{W}$ and any position $v\notin\mathcal{S}$, the incremental gain of the objective in $\mathcal{P}_3$ is given by}
	\begin{equation}
		{\blue \begin{aligned}
			\Delta(v\mid\mathcal{S})
			&=c(\mathcal{S}\cup\{v\})-c(\mathcal{S})\\
			&=\log_2\left(1+\rho\mathbf{w}^{H}(v)\mathbf{A}_{\mathcal{S}}\mathbf{w}(v)\right).
		\end{aligned}}
		\label{eq:om}
	\end{equation}
	{\blue Here,}
	\begin{equation}
		{\blue \mathbf{A}_{\mathcal{S}}=\left(\mathbf{I}_{K\!N_{\mathrm{U}}}+\rho\sum_{u\in\mathcal{S}}\mathbf{w}(u)\mathbf{w}^{H}(u)\right)^{-1}.}\label{eq:A}
	\end{equation}
	{\blue Equivalently, the objective in $\mathcal{P}_3$ can be written as}
	\begin{equation}
		{\blue c(\mathcal{S})=\log_2\det\left(\mathbf{I}_{K\!N_{\mathrm{U}}}+\rho\sum_{u\in\mathcal{S}}\mathbf{w}(u)\mathbf{w}^{H}(u)\right).}
	\end{equation}
	{\blue The per-position channel vector is}
	\begin{equation}
		{\blue \mathbf{w}(v)
		= \mathbf{G}^{H}
		\mathbf{\Sigma}^{H}
		\mathbf{f}(v),}
		\label{eq:w_{rm}}
	\end{equation}
	{\blue where $\mathbf{f}(v)$ denotes the $v$-th column of $\bar{\mathbf{F}}$ in \eqref{eq:codebook F}.}
	\label{lemma:1}
\end{lemma}
\begin{IEEEproof}
	First, based on \eqref{eq:w_{rm}}, we reorganize the MI as
	\begin{equation}
		\begin{aligned}
			c(\mathbf{v}) \;=\;& \log_2\det\left(\mathbf{I}_{N_{\mathrm{BS}}}
			+ \rho
			\mathbf{H}(\mathbf{v})
			\mathbf{H}^{H}(\mathbf{v})\right)\\
			=\;& \log_2\det\left(\mathbf{I}_{N_{\mathrm{BS}}} 
			+ \rho
			\mathbf{W}(\mathbf{v})
			\mathbf{W}^{H}(\mathbf{v})\right)\\
			=\;& \log_2\det\left(\mathbf{I}_{K\!N_{\mathrm{U}}} 
			+ \rho
			\mathbf{W}^{H}(\mathbf{v})
			\mathbf{W}(\mathbf{v})\right)\\
			=\;& \log_2\det\left(
			\mathbf{I}_{{K\!N_{\mathrm{U}}}} 
			+ \rho
			\sum_{m=1}^{N_{\mathrm{BS}}} 
			\mathbf{w}(v_m)
			\mathbf{w}^{H}(v_m)
			\right),
		\end{aligned}
		\label{eq:sum-rate}
	\end{equation}
	where \(\mathbf{W}^{H}(\mathbf{v})
	= \left[\mathbf{w}(v_1),\mathbf{w}(v_2),\ldots,\mathbf{w}(v_{N_{\mathrm{BS}}})\right]\). {\blue For any $\mathcal{S}\subseteq\mathcal{W}$ and $v\notin\mathcal{S}$, the matrix determinant lemma gives}
	\begin{equation}
		\begin{aligned}
			{\blue c(\mathcal{S}\cup\{v\})}
			=& {\blue \log_2\det\left(\mathbf{A}_{\mathcal{S}}^{-1}+\rho\mathbf{w}(v)\mathbf{w}^{H}(v)\right)}\\
			=& {\blue c(\mathcal{S})
			+\log_2\left(1+\rho\mathbf{w}^{H}(v)\mathbf{A}_{\mathcal{S}}\mathbf{w}(v)\right),}
		\end{aligned}
		\label{eq:cm}
	\end{equation}
	{\blue Subtracting $c(\mathcal{S})$ from both sides yields \eqref{eq:om},} which completes the proof.
\end{IEEEproof}
\begin{remark}
	Lemma 1 provides a fundamental insight into how spatial DoFs govern MI gains in MA systems. Specifically, \eqref{eq:A} reveals that {\blue $\mathbf{A}_{\mathcal{S}}$} is essentially the projection onto the orthogonal complement of the subspace spanned by the selected antennas. Therefore, \eqref{eq:om} indicated that the incremental gain of each selection depends on its capability to project energy onto the remaining uncorrelated spatial dimensions in the scattering environment.
	\label{remark:3}
\end{remark}

{\blue Next, we establish the monotonicity and the diminishing-return property of the MI function $c$ over an arbitrarily selected position set, as stated in the following lemma.}
\begin{lemma}
	\label{lemma:2}
	{\blue For any selected set $\mathcal{S}\subseteq\mathcal{W}$ and any position $v\notin\mathcal{S}$, the incremental gain satisfies $\Delta(v\mid\mathcal{S})\geq 0$. Moreover, for any nested sets $\mathcal{S}\subseteq\mathcal{T}\subseteq\mathcal{W}$ and any position $v\notin\mathcal{T}$, it holds that}
	\begin{equation}
		{\blue \Delta(v\mid\mathcal{T})\leq\Delta(v\mid\mathcal{S}).}
		\label{eq:nested_diminishing_returns}
	\end{equation}
\end{lemma}
\begin{IEEEproof}
{\blue Since the matrix inverted in \eqref{eq:A} is positive definite, we have $\mathbf{A}_{\mathcal{S}}\succ\mathbf{0}$, and \eqref{eq:om} therefore yields $\Delta(v\mid\mathcal{S})\geq 0$, which establishes the monotonicity of $c$.}

{\blue It remains to prove the diminishing-return property. To this end, we first characterize how $\mathbf{A}_{\mathcal{S}}$ evolves when an arbitrary position $u\notin\mathcal{S}$ is added to $\mathcal{S}$. According to \eqref{eq:A}, the relation between $\mathbf{A}_{\mathcal{S}}$ and $\mathbf{A}_{\mathcal{S}\cup\{u\}}$ is given by the rank-one update}
\begin{equation}
	{\blue \mathbf{A}_{\mathcal{S}\cup\{u\}}}
	= \left({\blue \mathbf{A}_{\mathcal{S}}^{-1}}
	+ \rho{\blue \mathbf{w}(u)\mathbf{w}^{H}(u)}
	\right)^{-1}.
	\label{eq:AS_update}
\end{equation}
{\blue By applying the Woodbury matrix identity, $\mathbf{A}_{\mathcal{S}\cup\{u\}}$ can be equivalently rewritten as}
\begin{equation}
	{\blue \mathbf{A}_{\mathcal{S}\cup\{u\}}}
	={\blue \mathbf{A}_{\mathcal{S}}}
	-\frac{\rho{\blue \mathbf{A}_{\mathcal{S}}\mathbf{w}(u)\mathbf{w}^{H}(u)\mathbf{A}_{\mathcal{S}}}}
	{1+\rho{\blue \mathbf{w}^{H}(u)\mathbf{A}_{\mathcal{S}}\mathbf{w}(u)}}.
	\label{eq:WI}
\end{equation}
{\blue Because $\rho\mathbf{w}(u)\mathbf{w}^{H}(u)\succeq\mathbf{0}$, it follows from \eqref{eq:AS_update} that $\mathbf{A}_{\mathcal{S}\cup\{u\}}^{-1}\succeq\mathbf{A}_{\mathcal{S}}^{-1}$, and hence $\mathbf{A}_{\mathcal{S}\cup\{u\}}\preceq\mathbf{A}_{\mathcal{S}}$ since matrix inversion reverses the L\"owner ordering. Therefore, any fixed position $v\notin\mathcal{S}\cup\{u\}$ satisfies}
\begin{equation}
	{\blue \mathbf{w}^{H}(v)\mathbf{A}_{\mathcal{S}\cup\{u\}}\mathbf{w}(v)
	\leq\mathbf{w}^{H}(v)\mathbf{A}_{\mathcal{S}}\mathbf{w}(v),}
\end{equation}
{\blue which, together with the monotonicity of the logarithm in \eqref{eq:om}, leads to}
\begin{equation}
	{\blue \Delta(v\mid\mathcal{S}\cup\{u\})\leq\Delta(v\mid\mathcal{S}).}
\end{equation}
{\blue For any nested pair $\mathcal{S}\subseteq\mathcal{T}$ and any position $v\notin\mathcal{T}$, inserting the positions of $\mathcal{T}\setminus\mathcal{S}$ one at a time and chaining this inequality yields}
\begin{equation}
	{\blue \Delta(v\mid\mathcal{T})\leq\Delta(v\mid\mathcal{S}),}
\end{equation}
{\blue so the incremental gain of the same position is non-increasing for every pair of nested selected sets, which is exactly the diminishing-return form \eqref{Equation.2} of submodularity. Therefore, the objective function $c$ is monotone submodular.}
\end{IEEEproof}

Lemma~\ref{lemma:2} establishes the diminishing-return property of the incremental gain {\blue of a given position over nested selected sets}. This property directly yields the following structural characteristics of the MA design problem.

\begin{proposition}
	The diminishing-return property implies that the objective function of the MA-enabled MI maximization problem is both submodular and monotonic.
	\label{proposition:1}
\end{proposition}
\begin{remark}
	Lemma~\ref{lemma:2} reveals the physical essence of how the MA affects the wireless channel. Specifically, \eqref{eq:WI} describes the relationship concerning the orthogonal complement of the subspace spanned by the antennas {\blue before and after adding a new antenna position to an arbitrarily selected set.} When only {\blue a} few antennas are deployed, {\blue $\mathbf{A}_{\mathcal{S}}$} approaches {\blue the} identity {\blue matrix}, indicating {\blue an} abundance of DoFs that {\blue allows} new antennas to capture strong uncorrelated multipath components. As {\blue the selected set grows}, {\blue $\mathbf{A}_{\mathcal{S}}$} contracts because prior antennas occupy dominant scattering paths, forcing subsequent additions to extract energy from diminishing orthogonal subspaces. {\blue Therefore}, as more MA positions are determined, the remaining spatial DoFs become progressively more constrained, and the incremental gain diminishes accordingly.
	\label{remark:4}
\end{remark}

Beyond revealing the physical mechanism of how the MA influences the wireless channel, submodularity represents a highly desirable property in integer and combinatorial optimization. \textcolor{black}{Critically, for monotone submodular objective functions, the greedy algorithm guarantees a solution achieving at least $(1-1/e)$ of the optimal value \cite{sviridenko2017optimal} with cardinality constraints. However, in our case, the minimum spacing constraint $|i-j| \ge N_{\mathrm{D}}$ introduces additional dependence to the feasible region.}
\begin{lemma}
	Consider the independence system defined by
	\begin{equation}
		\mathcal{I} \triangleq \left\{ \mathcal{S} \subseteq \mathcal{W} : |i - j| \geq N_{\mathrm{D}},\ \forall i,j \in \mathcal{S} \right\}.
	\end{equation}
	Then, for \( N_{\mathrm{D}} \geq 2 \) and \( N_{\mathrm{BS}} \geq 2 \), the system \( \mathcal{I} \) forms a \( k \)-system with \( k = 2 \).
	\label{lemma:3}
\end{lemma}
\begin{IEEEproof}[Proof]
	The non-emptiness and heredity properties are readily verified. We therefore focus on proving $k=2$. For any subset $\mathcal{A} \subseteq \mathcal{W}$ and any two maximal independent sets $\mathcal{I}, \mathcal{J} \subseteq \mathcal{A}$, since $\mathcal{J}$ is maximal in $\mathcal{A}$, for any point $y \in \mathcal{A}$, there exists $x \in \mathcal{J}$ such that $|y - x| < N_{\mathrm{D}}$. Thus, defining the neighborhood of $x$ with radius $N_{\mathrm{D}}$ within $\mathcal{A}$ as $\mathcal{B}_x = \left\{ y \in \mathcal{A} : |y - x| < N_{\mathrm{D}} \right\}$, we conclude that $\mathcal{A}$ is contained in the union of all $\mathcal{B}_x$, i.e., $\mathcal{A} \subseteq \bigcup_{x \in \mathcal{J}} \mathcal{B}_x$.
	
	Since all points in $\mathcal{I}$ are at least $N_{\mathrm{D}}$ apart, and $\mathcal{B}_x$ has diameter $2N_{\mathrm{D}}-1$, at most 2 points from $\mathcal{I}$ can lie in $\mathcal{B}_x$. Therefore, we can obtain
	\begin{equation}
		\left|\mathcal{I}\right|=\sum_{x\in \mathcal{J}}\left|\mathcal{I}\cap \mathcal{B}_x\right|\leq 2\left|\mathcal{J}\right|,
	\end{equation}
	which completes the proof.
\end{IEEEproof}

For monotone submodular maximization over a $k$-system constraint, the greedy algorithm guarantees a solution achieving at least \textcolor{black}{$(1+k)^{-1}$ of the optimal value. With $k=2$, this yields a performance guarantee of \textcolor{black}{at least} $\textcolor{black}{1/3}$.} While this bound provides a worst-case theoretical guarantee, our simulations in Sec.~VI {\blue will show that the MI gain achieved by the DCSPS typically reaches at least 90\% of the optimal solution's MI gain over fixed antennas.} This {\blue strong empirical result} further reinforces how submodularity endows low-complexity \textcolor{black}{DCSPS} algorithms with performance guarantees. Given the prohibitively high computational complexity of BnB methods, this combination of {\blue a theoretical guarantee and strong empirical MI results} is of paramount significance for real-world applications. The following subsection discusses the proposed \textcolor{black}{DCSPS} AIV design scheme.


\subsection{\textcolor{black}{Proposed \textcolor{black}{DCSPS} AIV Design Scheme}}
Based on the proved monotone submodularity property, we provide a \textcolor{black}{DCSPS} AIV design scheme for the proposed framework as a low-complexity scheme with an acceptable performance guarantee.

Specifically, we initialize an empty set for the selected MA positions $\mathcal{S} \leftarrow \emptyset$, indicating no antennas have yet been deployed. The iteration counter is set to $m \leftarrow 0$, while all $N_{\mathrm{S}}$ admissible locations on the antenna motion surface form the candidate pool $\mathcal{W}$. The algorithm will progressively select $N_{\mathrm{BS}}$ positions from $\mathcal{W}$ while strictly maintaining the minimum spacing constraint $|u - v| \geq N_{\mathrm{D}}$ between any two MAs to prevent electromagnetic coupling.

At each deployment iteration, the counter increases to $m \leftarrow m + 1$, representing the current number of MAs to be positioned. For every candidate location $v \in \mathcal{W} \setminus \mathcal{S}$ that satisfies the minimum distance requirement relative to all selected positions in $\mathcal{S}$, the communication system calculates the MI gain that would result from adding an MA at $v$. This incremental gain is quantified as
\begin{equation}
	\Delta(v | \mathcal{S}) = c(\mathcal{S} \cup \{v\}) - c(\mathcal{S}),
\end{equation}
where 
\begin{equation}
	c(\mathcal{S}) = \log_2 \det \left( \mathbf{I}_{\left|\mathcal{S}\right|} + \rho \mathbf{H}_{\mathcal{S}} \mathbf{H}_{\mathcal{S}}^H \right)\label{objr}
\end{equation}
represents the current MI of the partially configured antenna array. Here, $\mathbf{H}_{\mathcal{S}}$ is constructed from the predefined channel codebook $\mathbf{\bar{H}}$, which contains precomputed channel responses for all candidate locations.

The algorithm then identifies the optimal deployment position by solving
\begin{equation}
	v^* = \arg \max_{v} \Delta(v | \mathcal{S}),\label{greedyP}
\end{equation}
which physically corresponds to selecting the location that maximizes projection onto uncorrelated spatial dimensions in the scattering environment. This position $v^*$ is then added to the selected set via $\mathcal{S} \leftarrow \mathcal{S} \cup \{v^*\}$, triggering mechanical actuation to physically relocate the MA to the new position on the motion platform.

The entire process repeats until the counter reaches $m = N_{\mathrm{BS}}$, completing the full antenna array configuration. {\blue The number of iterations is determined in advance by $N_{\mathrm{S}}$ and $N_{\mathrm{BS}}$, so the procedure has a predetermined stopping point and does not involve a convergence process.} Throughout this sequential deployment, each iteration's decision directly corresponds to maximizing the exploitation of the remaining spatial DoFs, with the submodularity property ensuring that {\blue the }\textcolor{black}{DCSPS} {\blue scheme captures} progressively diminishing yet system-optimal MI gains. The overall low-complexity \textcolor{black}{DCSPS} scheme is summarized in \textbf{Algorithm \ref{alg1}}.
Beyond the $\textcolor{black}{1/3}$ performance guarantee afforded by submodularity, we now provide a physical explanation for the high MI attainable by the proposed \textcolor{black}{DCSPS} method.
\begin{remark}
	As indicated by \eqref{objr} and \eqref{greedyP}, the proposed \textcolor{black}{DCSPS} scheme operates on the principle of maximizing the incremental gain. This strategy aims to project as much energy as possible onto the remaining uncorrelated spatial dimensions within the scattering environment, thereby maximizing the utilization of the residual uncorrelated spatial DoFs. This directly enhances the MI.
	\label{remark:5}
\end{remark}
\begin{algorithm}[t]
	\caption{\textcolor{black}{DCSPS} Scheme with Perfect CSI}
	\label{alg1}
	\begin{algorithmic}[1]
		\REQUIRE $\mathbf{\Sigma}$, $\sigma^2$, $N_{\mathrm{U}}$, $N_{\mathrm{BS}}$, $N_{\mathrm{S}}$, $N_{\mathrm{D}}$, $\{ \theta_l \}_{l=1}^L$, $\{ \phi_l \}_{l=1}^L$, and $\lambda$
		\ENSURE $\mathcal{S}$ 
		\STATE Construct $\mathbf{G}$, $\bar{\mathbf{H}}$, $\bar{\mathbf{F}}$ via \eqref{Fk}-\eqref{g}, \eqref{eq:codebook channel}-\eqref{eq:codebook Fk};
		\STATE Initialize $\mathcal{S} \gets \emptyset$;
		\FOR{$m=1$ \TO $N_{\mathrm{BS}}$}
		\STATE $\mathrm{best\_value} \gets -\infty$, $\mathrm{best\_index} \gets \emptyset$;
		\FOR{$j=1$ \TO $N_{\mathrm{S}}$}
		\STATE $\mathrm{candidate\_valid} \gets \mathsf{TRUE}$;
		\IF{$\mathcal{S} \neq \emptyset$}
		\IF{$\min\limits_{s \in \mathcal{S}} |s - j| \leq N_{\mathrm{D}}$}
		\STATE $\mathrm{candidate\_valid} \gets \mathsf{FALSE}$; 
		\ENDIF
		\ENDIF
		\IF{$\mathrm{candidate\_valid}$}
		\STATE $\mathcal{S}_{\text{temp}} \gets \mathcal{S} \cup \{j\}$;
		\STATE Update $\mathbf{H}_{\mathcal{S}_{\text{temp}}}$ based on \eqref{eq:selected H};
		\STATE Obtain $c\left(\mathcal{S}_{\text{temp}}\right)$ based on \eqref{objr}; 
		\IF{$c\left(\mathcal{S}_{\text{temp}}\right) > \mathrm{best\_value}$}
		\STATE $\mathrm{best\_value} \gets c\left(\mathcal{S}_{\text{temp}}\right)$;
		\STATE $\mathrm{best\_index} \gets j$;
		\ENDIF
		\ENDIF
		\ENDFOR
		\IF{$\mathrm{best\_index} = \emptyset$}
		\STATE \textbf{error} ``No valid candidate found'';
		\ELSE
		\STATE $\mathcal{S} \gets \mathcal{S} \cup \{\mathrm{best\_index}\}$;
		\STATE Sort $\mathcal{S}$ in a descending order;
		\ENDIF
		\ENDFOR
	\end{algorithmic}
\end{algorithm}
\newcounter{TempEqCnt}                         
\setcounter{TempEqCnt}{\value{equation}} %
\setcounter{equation}{47}                           %
\begin{figure*}[!b]
	\hrulefill
	\begin{equation}
		\begin{aligned}
			\mathcal{P}_{5}:\,
			\max_{\mathcal{S}} &\quad c_{\mathrm{exp}}(\mathcal{S})\triangleq \mathbb{E}_{\Delta \mathbf{H}} \left\{
			\log_2 \det \left( 
			\mathbf{I}_{\left|\mathcal{S}\right|} 
			+ \rho \left(\hat{\mathbf{H}}_\mathcal{S} 
			+ \Delta \mathbf{H}_{\mathcal{S}}\right) \left(\hat{\mathbf{H}}_\mathcal{S} 
			+ \Delta \mathbf{H}_{\mathcal{S}}\right)^{H} \right) 
			\right\}\\
			\text{s.t.} 
			& \quad \text{C1:}\ \mathcal{S} \subseteq \mathcal{W}, \quad \text{C2:}\left|\mathcal{S}\right| = N_{\mathrm{BS}}, \quad \text{C3:} \left|i-j\right| \ge N_{\mathrm{D}}, \;\forall i,j \in \mathcal{S}.
		\end{aligned}
		\label{eq:P5}
	\end{equation}
\end{figure*}
\setcounter{equation}{\value{TempEqCnt}}

\subsection{Complexity Analysis}

The proposed \textcolor{black}{DCSPS} scheme exhibits a polynomial complexity of $\mathcal{O}(N_{\mathrm{S}} N_{\mathrm{BS}}^4)$, governed by the matrix operations in the iterative updates. Conversely, the BnB scheme incurs a prohibitive exponential cost of $\mathcal{O}(N_{\mathrm{eff}}^{N_{\mathrm{BS}}})$, driven by the exhaustive search over the effective position set $N_{\mathrm{eff}}$. By coupling this computational efficiency with a rigorous \textcolor{black}{1/3} approximation guarantee, the \textcolor{black}{DCSPS} approach avoids the intractability of combinatorial methods, establishing itself as a viable solution for practical implementations.

In the next section, we will extend the discussion to robust AIV design under imperfect CSI, highlighting how the submodularity property still applies in that scenario.

\section{Submodular Optimization with Imperfect CSI}
In practice, the perfect CSI assumption cannot be guaranteed due to {channel} estimation errors, limited pilot transmission, or feedback quantization \cite{9180053}. Consequently, it is crucial to discuss robust AIV design under imperfect CSI assumption.
In this section, we first present a commonly-used statistical channel estimation error model and then discuss whether the submodularity can be maintained under imperfect CSI. Finally, we propose an R-DCSPS design scheme that effectively exploits these conclusions.

\begin{figure*}[!t]
	\centering
	\begin{minipage}{2.75in}
		\centering
		\includegraphics[width=2.7in]{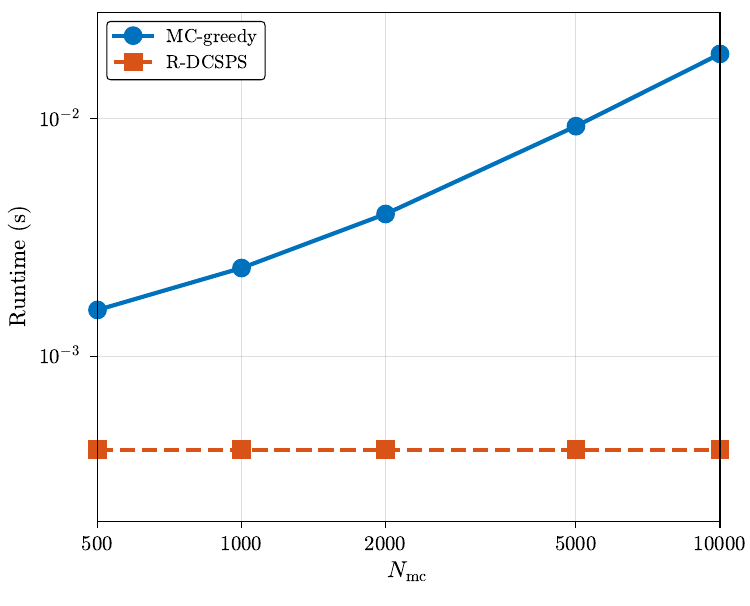}\\[-2pt]
		{\footnotesize (a)}
	\end{minipage}\hspace{0.25in}
	\begin{minipage}{2.75in}
		\centering
		\includegraphics[width=2.7in]{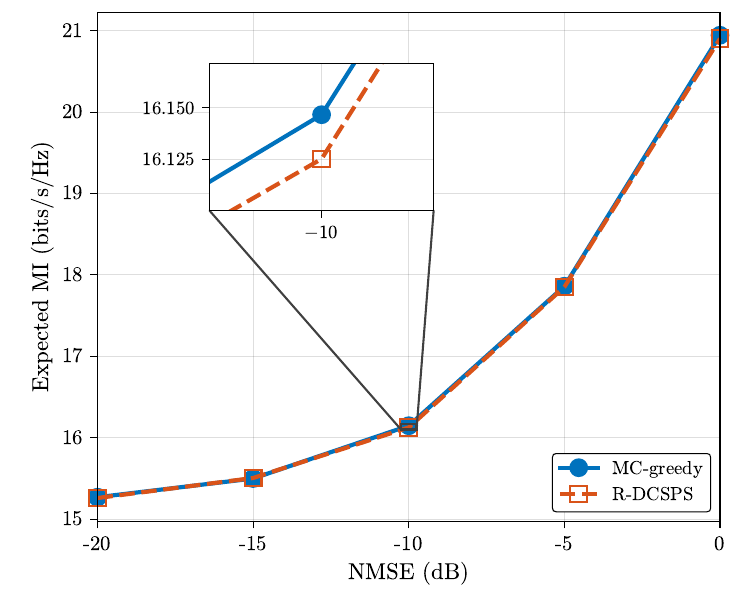}\\[-2pt]
		{\footnotesize (b)}
	\end{minipage}
	\caption{{\blue (a) Runtime v.s. $N_{\mathrm{mc}}$ at a target normalized mean square error (NMSE) of $-5$~dB; (b) Expected MI of the MC-greedy and R-DCSPS v.s. NMSE with $N_{\mathrm{mc}}=5000$. Results in both subfigures are averaged over 30 paired channel realizations with $N_{\mathrm{BS}}=4$, $N_{\mathrm{S}}=30$, $K=2$, $N_{\mathrm{U}}=2$, $L=20$, and $\mathrm{SNR}=30$~dB.}}
	\label{fig:jensen_surrogate}
\end{figure*}

\subsection{{Channel Estimation Error Model}}

We model the mismatch between the channel matrix \(\mathbf{H}\) and its estimate \(\hat{\mathbf{H}}\) via an additive error term, given by
\begin{equation}
	\mathbf{H} = \hat{\mathbf{H}} + \Delta \mathbf{H},
\end{equation}
where \(\hat{\mathbf{H}}\) denotes the estimated channel matrix, and \(\Delta \mathbf{H}\) is the estimation error matrix. 
This additive model is widely used in practical systems \cite{9180053}. For example, in time-division duplex (TDD) systems, the reciprocity-based channel estimation process may be corrupted by noise and interference during pilot transmissions, which are captured by \(\Delta \mathbf{H}\). By explicitly modeling \(\Delta \mathbf{H}\), we systematically incorporate error statistics into the subsequent optimization formulations, and further reveal the essence of the channel estimation error's effects.

Subsequently, we employ the statistical CSI error model \cite{9180053}, in which the channel estimation error \(\Delta \mathbf{H}\) follows the distribution $\operatorname{vec}\left(\Delta \mathbf{H}\right) \sim \mathcal{CN}\left(\mathbf{0}, \mathbf{\Psi}_{\mathrm{v}}\right)$, where the error covariance matrix $\mathbf{\Psi}_{\mathrm{v}}$ is a PSD matrix of dimension $KN_{\mathrm{U}}N_{\mathrm{BS}} \times KN_{\mathrm{U}}N_{\mathrm{BS}}$. Therefore, we have
	\begin{equation}
		\Delta \mathbf{H} \sim \mathcal{CN}\left(\mathbf{0}, \mathbf{\Psi}\right),
	\end{equation}
	where $\mathbf{\Psi}$ is a PSD matrix of dimension $N_{\mathrm{BS}} \times N_{\mathrm{BS}}$. The $(i,j)$-th element of $\mathbf{\Psi}$ is given by 
	\begin{equation}
		\mathbf{\Psi}_{i,j} = \sum_{k=1}^{KN_{\mathrm{U}}} \left[ \mathbf{\Psi}_{\mathrm{v}} \right]_{(k-1)N_{\mathrm{BS}} + i, \, (k-1)N_{\mathrm{BS}} + j}, \ i,j = 1,\ldots,N_{\mathrm{BS}}.
	\end{equation}
For a selected set $\mathcal{S}$, let $\Delta\mathbf{H}_{\mathcal{S}}$ denote the submatrix of $\Delta\mathbf{H}$ consisting of the rows indexed by $\mathcal{S}$, and let $\mathbf{\Psi}_{\mathcal{S}}\in\mathbb{C}^{\left|\mathcal{S}\right|\times\left|\mathcal{S}\right|}$ be the principal submatrix of $\mathbf{\Psi}$ consisting of the rows and columns indexed by $\mathcal{S}$. It then follows that $\mathbb{E}_{\Delta\mathbf{H}}\{\Delta\mathbf{H}_{\mathcal{S}}\Delta\mathbf{H}_{\mathcal{S}}^{H}\}=\mathbf{\Psi}_{\mathcal{S}}$. In particular, once the selection is complete, we have $\left|\mathcal{S}\right|=N_{\mathrm{BS}}$ and $\mathbf{\Psi}_{\mathcal{S}}=\mathbf{\Psi}$.

Since MA design is critically reliant on accurate CSI, directly employing the estimated CSI for MA optimization significantly degrades the MI in the presence of CSI inaccuracy. Addressing this challenge requires quantitatively characterizing the impact of such errors on the wireless channel. Subsequently, an optimization problem needs to be formulated to jointly address the dual objectives of mitigating the adverse effects introduced by CSI randomness and enhancing the MI within the MA design process.

\subsection{{Problem Formulation and Submodularity}}

After modeling channel estimation error accordingly, a common strategy is to consider the expected MI maximization problem for robust design \cite{9180053}. The robust optimization problem is given as \eqref{eq:P5} at the bottom of this page.

To address the stochastic nature of channel estimation errors, this formulation maximizes the expected {\blue MI $c_{\mathrm{exp}}(\mathcal{S})$ defined in \eqref{eq:P5}}. This metric explicitly captures the {\blue average MI} under CSI uncertainty, \textcolor{black}{thereby} optimizing long-term resource utilization efficiency \textcolor{black}{even in the presence of} imperfect channel conditions. Consequently, maximizing {\blue $c_{\mathrm{exp}}(\mathcal{S})$} provides a robust operational balance by emphasizing MI under the most frequent channel characteristics.

{\blue Although $c_{\mathrm{exp}}(\mathcal{S})$ is the desired optimization objective, the expectation over $\Delta\mathbf{H}$ does not admit a tractable closed-form expression. A direct greedy search would require Monte Carlo (MC) averaging for each candidate position in every selection step, which increases the computational complexity and makes the resulting search impractical for real-time position adaptation. We therefore apply Jensen's inequality to derive a deterministic closed-form surrogate. Such Jensen's inequality-based deterministic surrogates are also adopted in robust wireless communication design, as exemplified by \cite{li2023ergodic}.}

\setcounter{equation}{48}
According to Jensen's inequality, we have
\begin{equation}
	\begin{aligned}
		c_{\mathrm{exp}}(\mathcal{S})
		& \leq \log_2 \det \left( \mathbb{E}_{\Delta \mathbf{H}}\left\{\mathbf{I}_{\left|\mathcal{S}\right|} + \rho (\hat{\mathbf{H}}_{\mathcal{S}} + \Delta \mathbf{H}_{\mathcal{S}}) (\hat{\mathbf{H}}_{\mathcal{S}} + \Delta \mathbf{H}_{\mathcal{S}})^{H}\right\} \right) \\
		& \overset{(a)}{=} \log_2 \det \left( \mathbf{I}_{\left|\mathcal{S}\right|}  
		+ \rho\hat{\mathbf{H}}_{\mathcal{S}}\hat{\mathbf{H}}_{\mathcal{S}}^{H} 
		+ \rho\mathbb{E}_{\Delta \mathbf{H}}\{\Delta \mathbf{H}_{\mathcal{S}}\Delta \mathbf{H}_{\mathcal{S}}^{H}\} \right)\\
		& = \log_2 \det \left( \mathbf{I}_{\left|\mathcal{S}\right|}  
		+ \rho\hat{\mathbf{H}}_{\mathcal{S}} \hat{\mathbf{H}}_{\mathcal{S}}^{H} 
		+ \rho \mathbf{\Psi}_{\mathcal{S}} \right) \\
		& \overset{(b)}{=} \log_2 \det \left( \mathbf{I}_{\left|\mathcal{S}\right|}  
		+ \rho\hat{\mathbf{H}}_{\mathcal{S}}\hat{\mathbf{H}}_{\mathcal{S}}^{H} +\rho \mathbf{E}_{\mathcal{S}}\mathbf{E}_{\mathcal{S}}^H \right)\\
		& = \log _2 \operatorname{det}\left(\mathbf{I}_{\left|\mathcal{S}\right|} +\rho\left[\begin{array}{ll}
			\mathbf{E}_{\mathcal{S}} \ \ \hat{\mathbf{H}}_{\mathcal{S}}
		\end{array}\right]\left[\begin{array}{l}
			\mathbf{E}_{\mathcal{S}}^H \\
			\hat{\mathbf{H}}_{\mathcal{S}}^H
		\end{array}\right]\right)\\
		& = \log _2 \operatorname{det}\left(\mathbf{I}_{\left|\mathcal{S}\right|} +\rho \tilde{\mathbf{H}}_{\mathcal{S}}\tilde{\mathbf{H}}_{\mathcal{S}}^H\right){\blue \triangleq c_{\mathrm{r}}(\mathcal{S})},
	\end{aligned}
	\label{33}
\end{equation}
where step $(a)$ follows from
\begin{equation}
	\mathbb{E}_{\Delta\mathbf{H}}\{\hat{\mathbf{H}}_{\mathcal{S}}^{H}\Delta\mathbf{H}_{\mathcal{S}}\}
	= \hat{\mathbf{H}}_{\mathcal{S}}^{H}\mathbb{E}_{\Delta\mathbf{H}}\{\Delta\mathbf{H}_{\mathcal{S}}\} 
	= \mathbf{0},
\end{equation}
and similarly $\mathbb{E}_{\Delta\mathbf{H}}\{\Delta\mathbf{H}_{\mathcal{S}}^{H}\hat{\mathbf{H}}_{\mathcal{S}}\} = \mathbf{0}$.

In addition, $(b)$ can be guaranteed since $\mathbf{\Psi}$ is a PSD matrix and thus admits the factorization $\mathbf{\Psi}=\mathbf{E}\mathbf{E}^{H}$ with $\mathbf{E}=\mathbf{\Psi}^{1/2}\in\mathbb{C}^{N_{\mathrm{BS}}\times N_{\mathrm{BS}}}$, where $\mathbf{E}_{\mathcal{S}}\in\mathbb{C}^{\left|\mathcal{S}\right|\times N_{\mathrm{BS}}}$ denotes the submatrix of $\mathbf{E}$ consisting of the rows indexed by $\mathcal{S}$, so that $\mathbf{\Psi}_{\mathcal{S}}=\mathbf{E}_{\mathcal{S}}\mathbf{E}_{\mathcal{S}}^{H}$ holds for every selected set. 
Here, we reveal the inherent nature of the impact of channel estimation errors. First, we define the virtual channel matrix $\tilde{\mathbf{H}}_{\mathcal{S}}\in \mathbb{C}^{\left|\mathcal{S}\right| \times (N_{\mathrm{BS}}+K N_{\mathrm{U}})}$ as 
\begin{equation}
	\tilde{\mathbf{H}}_{\mathcal{S}}=\left[\mathbf{E}_{\mathcal{S}} \ \ \hat{\mathbf{H}}_{\mathcal{S}}\right].\label{34}
\end{equation}
Based on our derivation in Sec.~IV, we see that the surrogate objective $c_{\mathrm{r}}(\mathcal{S})$ in \eqref{33} is a monotone submodular function. {\blue Therefore, the $1/3$ performance guarantee of the R-DCSPS applies to the surrogate objective $c_{\mathrm{r}}(\mathcal{S})$ under the 2-system spacing constraint.}

{\blue To verify the rationale of adopting Jensen's inequality in \eqref{33}, we introduce a high-precision MC greedy baseline, denoted as MC-greedy, that directly optimizes the original expected MI objective $c_{\mathrm{exp}}(\mathcal{S})$ by evaluating each candidate's incremental gain over $N_{\mathrm{mc}}$ MC realizations of $\Delta\mathbf{H}$.}

{\blue Fig.~\ref{fig:jensen_surrogate}(a) quantifies the computational cost of directly optimizing $c_{\mathrm{exp}}(\mathcal{S})$. The runtime of the MC-greedy baseline drastically increases with $N_{\mathrm{mc}}$, whereas the runtime of the R-DCSPS remains independent of $N_{\mathrm{mc}}$. The MC-greedy baseline is $23.1$ times slower than the R-DCSPS when $N_{\mathrm{mc}}=5000$ and $46.5$ times slower when $N_{\mathrm{mc}}=10{,}000$. This is because the closed-form Jensen's surrogate $c_{\mathrm{r}}(\mathcal{S})$ eliminates the MC averaging required to evaluate the expected MI incremental gains. On the other hand, }
{\blue Fig.~\ref{fig:jensen_surrogate}(b) quantifies the performance loss caused by replacing $c_{\mathrm{exp}}(\mathcal{S})$ with $c_{\mathrm{r}}(\mathcal{S})$. Over the entire NMSE range from $-20$~dB to $0$~dB, the largest observed expected MI loss of the R-DCSPS relative to the high-precision MC-greedy baseline is $0.041$~bits/s/Hz, which is $0.20\%$ of the expected MI achieved by the MC-greedy baseline at the corresponding point. Hence, the Jensen's surrogate $c_{\mathrm{r}}(\mathcal{S})$ removes the dominant computational overhead of MC averaging while incurring negligible MI loss relative to direct MC optimization of $c_{\mathrm{exp}}(\mathcal{S})$.}

\begin{remark}
	The virtual channel matrix in \eqref{34} fundamentally establishes a virtual user paradigm where channel estimation errors are equivalent to introducing additional virtual users into the system. These phantom entities actively compete for spatial resources, forcing the BS to utilize its $N_{\mathrm{BS}}$ antennas to serve an expanded set comprising both the $K N_{\mathrm{U}}$ data streams of the actual users and $N_{\mathrm{BS}}$ virtual streams introduced by the channel estimation error. This competition partitions the finite spatial DoFs between the two groups, inevitably reducing the effective DoFs accessible for genuine data transmission. Consequently, the communication resources that would otherwise serve real users are diverted to virtual counterparts, directly constraining the system's MI regardless of specific channel realizations. This model explicitly quantifies how estimation uncertainty consumes spatial DoFs and {\blue reduces the MI}. The mechanism ensures every antenna deployment intrinsically balances MI against error vulnerability, transforming statistical imperfections into manageable spatial constraints while preserving the $\textcolor{black}{1/3}$ performance guarantee of the R-DCSPS scheme for the surrogate objective $c_{\mathrm{r}}(\mathcal{S})$.
	\label{remark:6}
\end{remark}
\subsection{{\textcolor{black}{R-DCSPS} Design Scheme under Imperfect CSI}}
The submodular optimization framework established for perfect CSI in Sec.~IV-C remains fully applicable under imperfect CSI conditions. The algorithm retains identical operational logic: Sequential position selection from discrete codebook $\mathcal{W}$ while enforcing minimum spacing constraints $|i-j| \geq N_{\mathrm{D}}$ to prevent electromagnetic coupling between MAs. The fundamental adaptation occurs in the incremental gain calculation, which now incorporates statistical channel uncertainty through a reformulated physical model.

During each iteration, the system evaluates candidate positions $v \in \mathcal{W} \setminus \mathcal{S}$ satisfying antenna spacing constraints. For each valid $v$, it computes the robust incremental gain $\Delta_r(v \mid \mathcal{S})$ as follows
\begin{equation}
	\Delta_{\mathrm{r}}(v\mid\mathcal{S}) \;=\;
		c_{\mathrm{r}}\left(\mathcal{S}\cup\{v\}\right)-c_{\mathrm{r}}\left(\mathcal{S}\right),
	\label{eq:delta_r}
\end{equation}
where the surrogate objective $c_{\mathrm{r}}(\mathcal{S})$ is defined in \eqref{33}.
The position maximizing this gain is then selected, triggering physical MA relocation via mechanical actuators. {\blue The R-DCSPS scheme follows the same finite control flow and stopping criterion $|\mathcal{S}|=N_{\mathrm{BS}}$ as \textbf{Algorithm~1}. It only replaces the incremental gain $\Delta(v\mid\mathcal{S})$ with the closed-form robust incremental gain $\Delta_{\mathrm{r}}(v\mid\mathcal{S})$ in \eqref{eq:delta_r}, so it requires neither runtime sampling nor a convergence tolerance. Since it performs the same candidate scan as the DCSPS scheme, the R-DCSPS scheme retains the polynomial complexity order of $\mathcal{O}(N_{\mathrm{S}}N_{\mathrm{BS}}^{4})$, while the wider virtual channel and the one-time factorization of the error covariance $\mathbf{\Psi}$ introduce only a constant-factor overhead.}

\section{Simulation Results}
Numerical results based on MC simulations are presented in this section to validate the effectiveness of the proposed algorithms. The multipath channel model is considered. Specifically, for the $k$-th user, the predefined codebook channel matrix in \eqref{eq:codebook channel} is generated by
\begin{equation}
	\bar{\mathbf{H}}_k=\sum_{l=1}^L\alpha_{l,k}\bar{\mathbf{f}}_{k}^{*}(\phi_{l,k})\bar{\mathbf{g}}_{k}^{T}(\theta_{l,k}),
\end{equation}
where
\begin{equation}
	\bar{\mathbf{f}}_{k}(\phi_{l,k})=\left[ 1, e^{\jmath\frac{2\pi A}{N_{\mathrm{S}}{\lambda}}\cos\phi_{l,k}},
	\ldots,\ 
	e^{\jmath\frac{2\pi A}{N_{\mathrm{S}}{\lambda}}(N_{\mathrm{S}}-1)\cos\phi_{l,k}} \right]^{T},
\end{equation}
and
\begin{equation}
	\bar{\mathbf{g}}_{k}(\theta_{l,k})=\left[ 1, e^{\jmath\pi\cos\theta_{l,k}},
	\ldots,\ 
	e^{\jmath\pi(N_{\mathrm{U}}-1)\cos\theta_{l,k}} \right]^{T}.
\end{equation}

The complex path gain $\alpha_{l,k}$ follows an independent and identically distributed $\mathcal{CN}(0, L_0D_k^{-\gamma})$ distribution, where $L_0=1$, $D_k$, and $\gamma=2.2$ denote the large-scale fading at reference distance $D_0=1\ \mathrm{m}$, distance from the BS to the $k$-th user, and the path loss exponent, respectively \cite{10437926}. In simulations, we assume that $\left\{D_k\right\}_{k=1}^K$ are uniformly distributed between $20\ \mathrm{m}$ to $100\ \mathrm{m}$. Meanwhile, a typical minimum distance constraint $N_{\mathrm{D}}$ in \eqref{P0} constraints can be set as
\begin{equation}
	{\blue N_{\mathrm{D}}=\left\lceil \frac{N_{\mathrm{S}}D}{A}\right\rceil,}
	\label{eq:ND_setting}
\end{equation}
where $D=\frac{\lambda}{2}$ \cite{10243545}, {\blue and the selection of this value shall be discussed in Sec. VI-A.} The carrier frequency is set to $5\text{ GHz}$, i.e., the wavelength is $\lambda=0.06\ \mathrm{m}$.
Here, $\left\lceil\cdot\right\rceil$ denotes the ceiling function, which yields the smallest integer greater than or equal to its argument. 
Unless otherwise stated, we set $K=2$, $L=20$, and the half-wavelength antenna spacing for each user. Both $\theta_{l,k}$ and $\phi_{l,k}$ are uniformly distributed in the range of $[0,\pi]$. All the results are obtained by averaging over 1,000 randomly generated channel realizations. 

Numerical results of system MI are presented in this section. To demonstrate the MI gains, the following schemes are adopted for comparison throughout this section:
\begin{enumerate}
	\item `BnB Scheme' \cite{cheng2025exploitingmovableantennasmulticast}: We implement a BnB method that attains the highest possible MI as the upper bound. This approach is typically computationally expensive in practical applications, particularly when the feasible region is large.
	\item `DCSPS/R-DCSPS': The proposed DCSPS/R-DCSPS schemes according to Secs.~IV and V.
	\item `Benchmark': The conventional fixed half-wavelength antenna spacing is considered at the BS. {\blue Since its antenna positions, namely its only design variable, are fixed prior to channel estimation and hence independent of the CSI, this scheme yields identical performance under perfect and imperfect CSI.} 
	\item {\blue `Random': An antenna placement is randomly generated from the sets that satisfy the minimum-spacing constraint.}
	\item {\blue `Equidistant': The antennas are uniformly distributed over the enlarged aperture and mapped to the nearest candidate grids.}
	\item {\blue `Channel-Gain (CG) Greedy' \cite{abuzgaia2026fas}}: {\blue At each selection round, the position with the largest squared row norm of the channel matrix is selected among those satisfying the minimum-spacing constraint.}
\end{enumerate}

\subsection{Selection of the Minimum Antenna Spacing}
\begin{figure}[t]
	\centering
	\includegraphics[width=3.05in,height=2.493in]{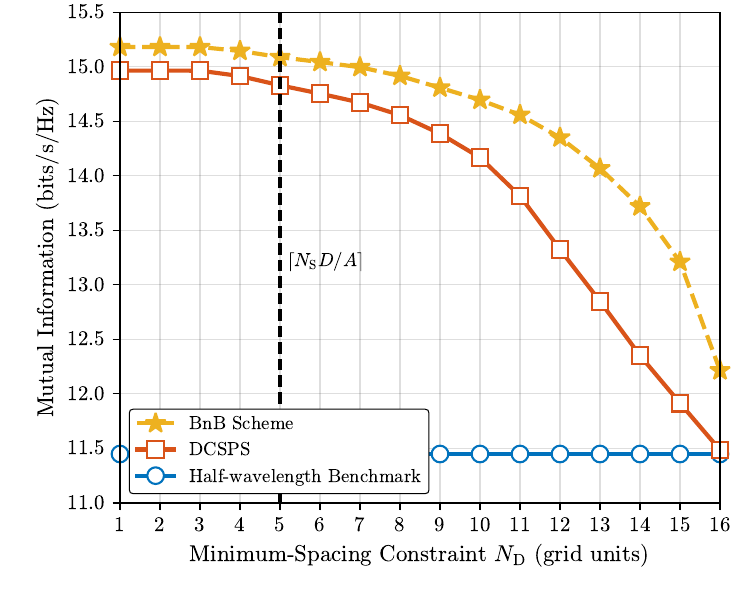}
	\caption{{\blue MI v.s. minimum-spacing constraint $N_{\mathrm{D}}$ under perfect CSI, with $A=2(N_{\mathrm{BS}}-1)\lambda$, $N_{\mathrm{BS}}=4$, $N_{\mathrm{U}}=2$, $K=2$, $N_{\mathrm{S}}=50$, and $\text{SNR}=30$~dB. The dashed vertical line marks the setting $\lceil N_{\mathrm{S}}D/A\rceil=5$ of \eqref{eq:ND_setting}.}}
	\label{fig:nd_sweep}
\end{figure}
{\blue To characterize the impact of the minimum-spacing constraint, we sweep $N_{\mathrm{D}}$ over its entire feasible range from $1$ to $N_{\mathrm{D}}^{\max}=\lfloor(N_{\mathrm{S}}-1)/(N_{\mathrm{BS}}-1)\rfloor=16$. Fig.~\ref{fig:nd_sweep} shows that both the BnB optimum and the MI of the DCSPS decrease monotonically as $N_{\mathrm{D}}$ increases. The monotonicity of the BnB optimum can be explained by the nested feasible sets. Specifically, let $\mathcal{F}(N_{\mathrm{D}})$ denote the family of cardinality-$N_{\mathrm{BS}}$ placements that satisfy the spacing constraint, and let $c^{\star}(N_{\mathrm{D}})=\max_{\mathcal{S}\in\mathcal{F}(N_{\mathrm{D}})}c(\mathcal{S})$ denote the corresponding optimum MI. Since $\mathcal{F}(N_{\mathrm{D}}+1)\subseteq\mathcal{F}(N_{\mathrm{D}})$, the optimum satisfies $c^{\star}(N_{\mathrm{D}}+1)\le c^{\star}(N_{\mathrm{D}})$. As a result, a larger $N_{\mathrm{D}}$ removes candidate placements and can only reduce the optimum MI, while the MI gain over the fixed half-wavelength benchmark gradually diminishes. For the DCSPS, increasing $N_{\mathrm{D}}$ contracts the feasible candidate set at every selection iteration, so the greedy search has less freedom to select positions that add complementary spatial dimensions. The attainable incremental MI gains are therefore reduced, which explains the decreasing MI trend observed in Fig.~\ref{fig:nd_sweep}.}

{\blue Although a smaller value of $N_{\mathrm{D}}$ leads to a higher MI, in this paper, we proposed to set $N_{\mathrm{D}}=\lceil N_{\mathrm{S}}D/A\rceil$, where $D=\lambda/2$. The half-wavelength spacing provides sufficiently low spatial correlation, allowing the selected antennas to exploit more spatial DoFs and achieve a higher MI. Therefore, the setting $D=\lambda/2$ adopted in \eqref{eq:ND_setting} is the typical half-wavelength choice widely used in the MIMO literature \cite{10906511,10243545}, and it gives $N_{\mathrm{D}}=5$ in this section unless otherwise specified. The dashed line therefore marks this physically motivated operating point rather than an MI-maximizing value or an exact transition point.}

\FloatBarrier

\subsection{Validation of Monotonicity and Submodularity}
\begin{figure}[t]
	\centering
	\includegraphics[width=3.05in,height=2.493in]{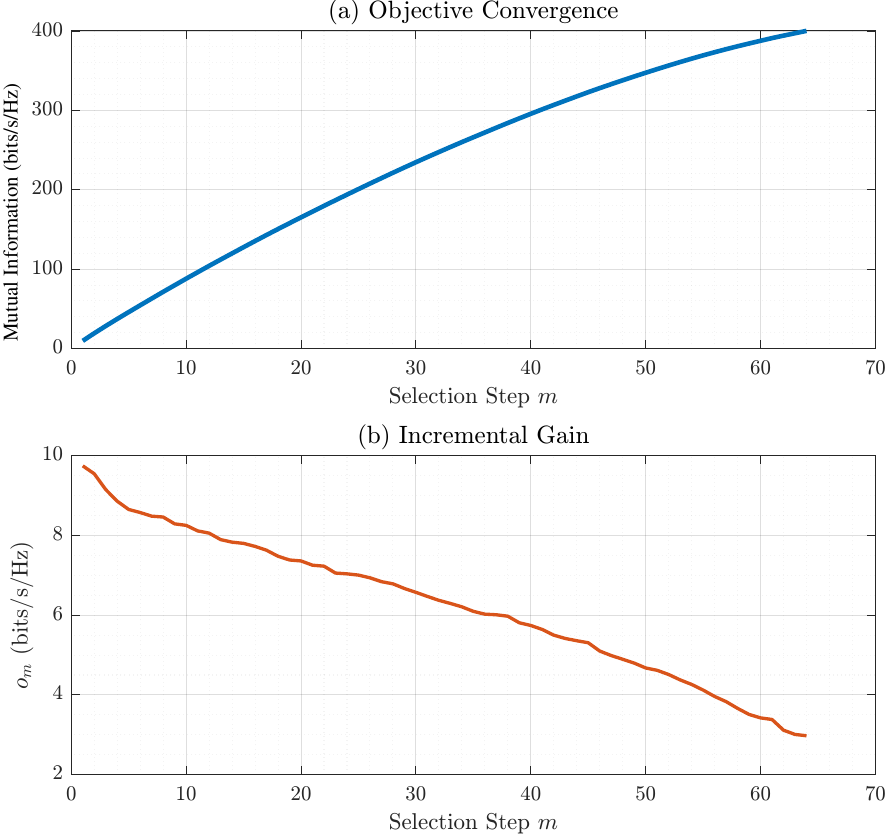}
	\caption{MI convergence and incremental MI gain during the greedy selection process under perfect CSI when $N_{\mathrm{BS}}=64$, $K = 8$, $N_{\mathrm{U}} = 8$, {$\text{SNR} = 30\ \text{dB}$}, and $N_{\mathrm{S}} = 100$.}
	\label{fig3}
\end{figure}

Fig. \ref{fig3} indicates that the MA-position design problem whose goal is to maximize the MI, is a monotone submodular optimization problem. In Fig. \ref{fig3}(a), the MI increases monotonically throughout the greedy selection process, confirming the monotonicity of the objective function. Meanwhile, Fig. \ref{fig3}(b) depicts the incremental MI gain obtained at each successive selection step, and its consistently declining curve exhibits a clear pattern of diminishing returns. Each new selection provides a diminishing incremental gain, empirically demonstrating the submodularity of the objective function.

\subsection{\texorpdfstring{MI}{MI} v.s. Transmit SNR}


\begin{figure}[t]
	\centering
	\includegraphics[width=3.1in,height=2.493in]{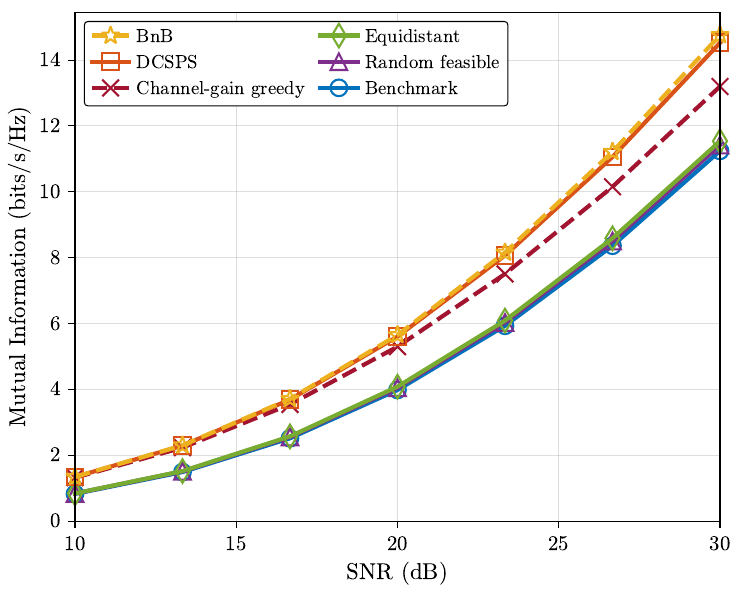}
	\caption{{\blue MI v.s. transmit SNR under the perfect CSI when $A=2(N_{\mathrm{BS}}-1)\lambda$ and $N_{\mathrm{S}}=30$.}}
	\label{fig:controlled_baselines_perfect}
\end{figure}

{\blue To isolate the contribution of the submodular MI criterion, we compare the DCSPS with the BnB, half-wavelength benchmark, CG greedy, random feasible placement, and equidistant placement schemes under perfect CSI. Fig.~\ref{fig:controlled_baselines_perfect} shows that the DCSPS consistently outperforms all low-complexity baselines and closely tracks the BnB upper bound. At $30$~dB, the DCSPS achieves $14.522$~bits/s/Hz, which is $3.280$~bits/s/Hz higher than the benchmark and $1.324$~bits/s/Hz higher than the CG greedy, while remaining only $0.237$~bits/s/Hz below the BnB optimum. Its superiority over the random feasible and equidistant placements shows that aperture enlargement and spacing feasibility alone are insufficient, while its advantage over the CG greedy baseline confirms that the additional gain arises from the submodular MI criterion rather than from implementation differences. These results demonstrate that the DCSPS achieves a substantial placement gain and near-optimal MI at polynomial complexity, thereby avoiding the combinatorial search required by the BnB.}

\begin{figure}[t]
	\centering
	\includegraphics[width=3.05in,height=2.493in]{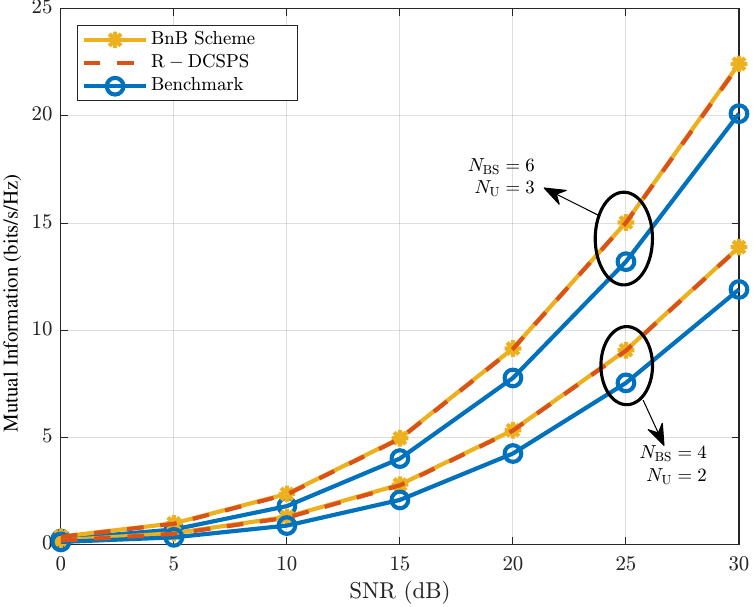}
	\caption{MI v.s. transmit SNR under imperfect CSI when $\text{NMSE}=-5\text{ dB}$, $A=2(N_{\mathrm{BS}}-1)\lambda$, and $N_{\mathrm{S}} = 30$.}
	\label{fig5}
\end{figure}

Fig. \ref{fig5} depicts the MI v.s. transmit SNR under imperfect CSI with an NMSE of $-5\text{ dB}$ in channel estimation. 
Compared to Fig. \ref{fig:controlled_baselines_perfect}, two key observations emerge: First, the MI gap of all schemes between perfect and imperfect CSI scenarios directly demonstrates the adverse impact of channel estimation errors on MA positioning. Nevertheless, the proposed \textcolor{black}{R-DCSPS} scheme achieves substantial MI gains over the fixed half-wavelength antenna array benchmark. This indicates that incorporating CSI uncertainty into the algorithm design effectively mitigates MI degradation due to estimation errors, preserving considerable MI gains in practical systems.

\begin{figure}[t]
	\centering
	\includegraphics[width=3.15in,height=2.493in]{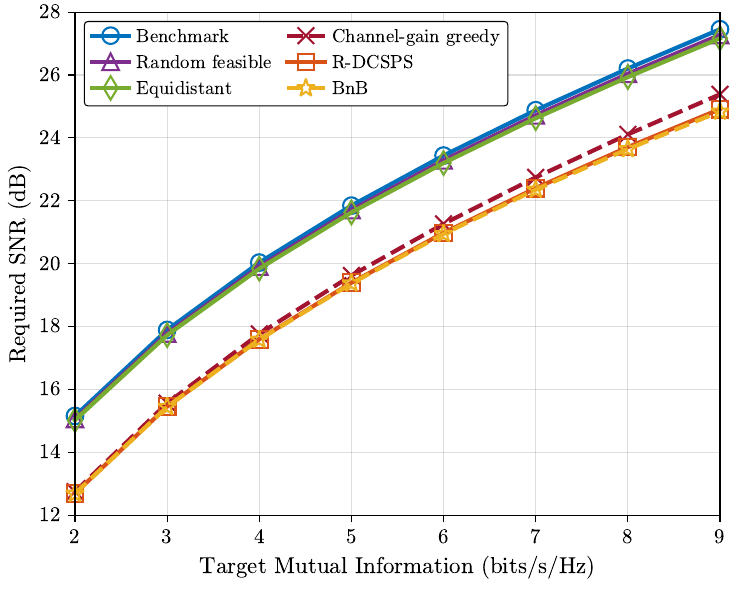}
	\caption{{\blue Required SNR v.s. target MI under the imperfect CSI when the NMSE is $-10$~dB, $A=2(N_{\mathrm{BS}}-1)\lambda$, and $N_{\mathrm{S}}=30$.}}
	\label{fig:required_snr_imperfect}
\end{figure}

{\blue Fig.~\ref{fig:required_snr_imperfect} quantifies the SNR required by each scheme to attain a target MI under imperfect CSI. At a target MI of $9$~bits/s/Hz, the R-DCSPS requires $24.913$~dB while saving $2.542$~dB, $2.363$~dB, $2.253$~dB, and $0.471$~dB relative to the benchmark, random feasible placement, equidistant placement, and CG greedy, respectively. The R-DCSPS and CG greedy require nearly the same SNR at the lowest target MI because the MI reduces to a channel-gain term to first order in the low-SNR regime. The SNR saving of the R-DCSPS over the CG greedy then increases from $0.054$~dB at $2$~bits/s/Hz to $0.471$~dB at $9$~bits/s/Hz, which shows that the benefit of the robust submodular MI criterion becomes more pronounced as the target MI increases. The BnB scheme requires only $0.008$~dB to $0.074$~dB less SNR than the R-DCSPS over the entire target range. The increasing SNR advantage over the CG greedy baseline, together with the small gap to the BnB, confirms the effectiveness of the R-DCSPS among the low-complexity schemes under imperfect CSI.}

\subsection{\texorpdfstring{{\blue Performance Gains with Practical Transceivers}}{Performance Gains in Practical Systems}}

{\blue Recall that, to facilitate a tractable AIV design problem, we assumed an isotropic transmission at the transmit side and an MMSE-SIC receiver to achieve the MI. To evaluate the performance gains of the proposed position design in practical systems, we consider both transmit covariance design at the users and receiver design at the BS. For the transmit-side study, we compare the DCSPS with the half-wavelength benchmark under isotropic transmission, uplink transmit zero-forcing (ZF), per-user singular value decomposition-based water-filling (SVD-WF) \cite[Ch.~7]{tse2005fundamentals}, and iterative water-filling (IWF) \cite{yu2004iterative}. For every non-isotropic DCSPS curve, the antenna positions and transmit covariance matrices are designed via AO, whereas the benchmark positions remain fixed and only the covariance matrices are updated. For the receiver-side study, we compare the DCSPS with the same benchmark under the MMSE-SIC, linear minimum mean-square error (LMMSE), and ZF receivers, all with the isotropic input.}

\begin{figure}[t]
	\centering
	\includegraphics[width=3.2in,height=2.493in]{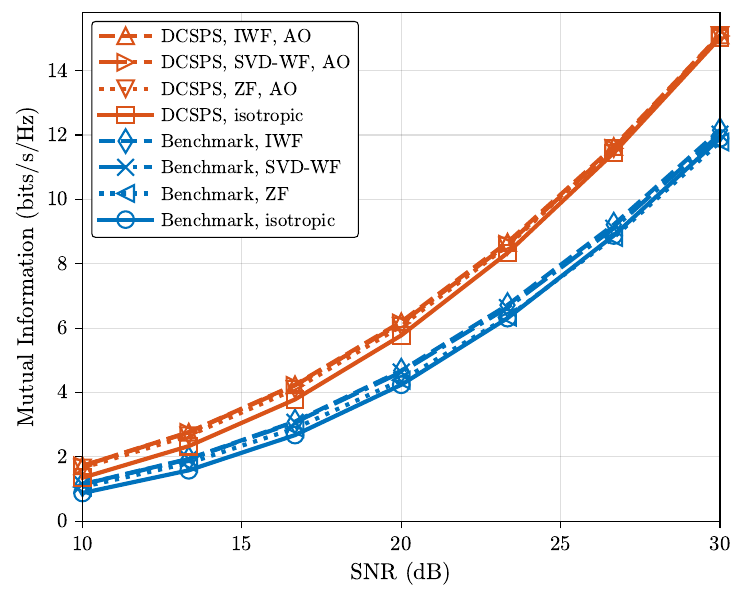}
	\caption{{\blue MI v.s. SNR under perfect CSI and non-isotropic transmissions with $N_{\mathrm{BS}}=4$, $N_{\mathrm{S}}=30$, $K=2$, $N_{\mathrm{U}}=2$, and $L=20$.}}
	\label{fig:covariance_design}
\end{figure}

{\blue Fig.~\ref{fig:covariance_design} shows that the MI achieved by the isotropic transmission and directional covariance designs is nearly identical. In addition, under the same covariance design, the DCSPS retains a significant MI gain over the corresponding half-wavelength benchmark. These observations indicate that the antenna positioning has a more profound impact on MI than covariance design and confirm the superiority of the proposed DCSPS when practical directional transmit covariance is incorporated. On the other hand, in Fig.~\ref{fig:receiver_design}, the proposed DCSPS consistently maintains a substantial sum-rate gain over the benchmark under each receiver. Quantitatively, at $30$~dB the DCSPS improves the benchmark by about $3.12$~bits/s/Hz under MMSE-SIC, $4.00$~bits/s/Hz under LMMSE, and $5.72$~bits/s/Hz under the ZF receiver. Relative to the corresponding benchmark rates, these gains amount to $26.3\%$, $42.9\%$, and $82.8\%$, respectively. This indicates that the proposed DCSPS achieves a significant improvement, which becomes even more prominent under practical receivers.}

\begin{figure}[t]
	\centering
	\includegraphics[width=3.05in,height=2.493in]{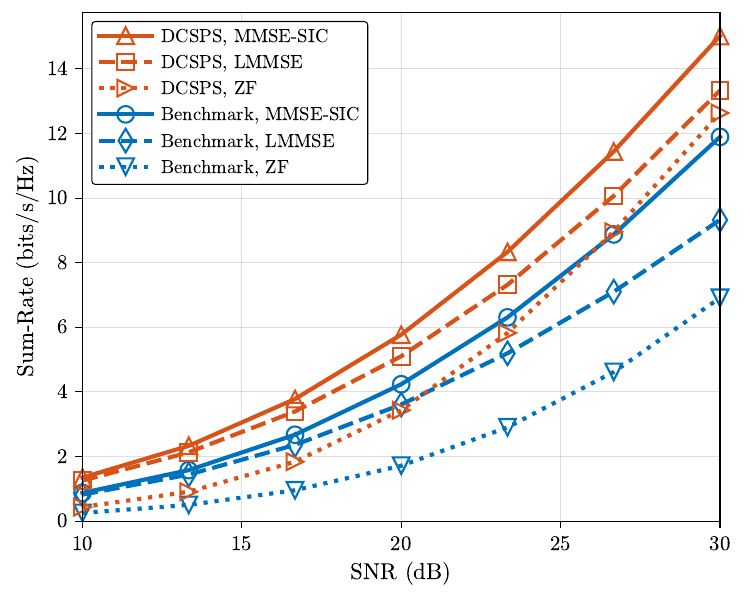}
	\caption{{\blue Sum-Rate v.s. SNR under perfect CSI and practical receivers, with $N_{\mathrm{BS}}=4$, $N_{\mathrm{S}}=30$, $K=2$, $N_{\mathrm{U}}=2$, and $L=20$.}}
	\label{fig:receiver_design}
\end{figure}

\subsection{\texorpdfstring{MI}{MI} v.s. Key System Parameters}

\begin{figure}[t]
	\centering
	\includegraphics[width=3.05in,height=2.493in]{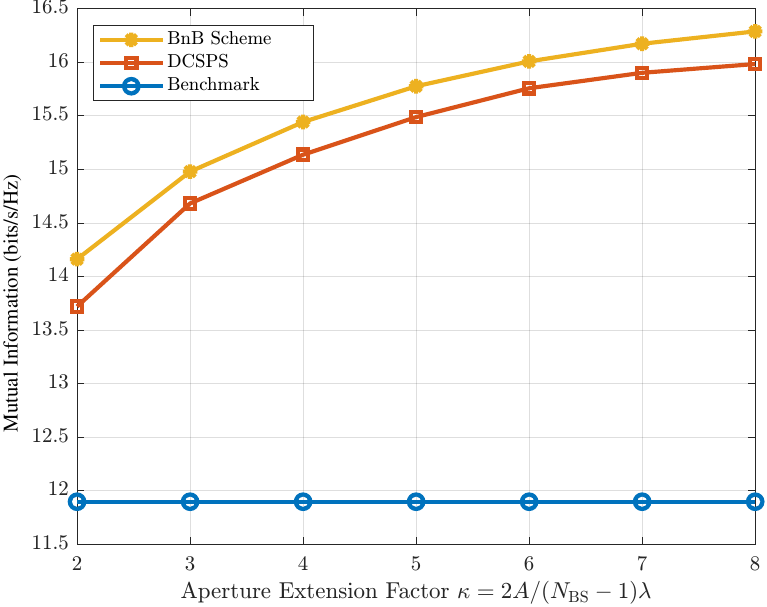}
	\caption{MI v.s. aperture under perfect CSI when $N_{\mathrm{BS}}=4$, $N_{\mathrm{U}}=2$, $\text{SNR} = 30\ \text{dB}$, and $N_{\mathrm{S}} = 50$.}
	\label{fig6}
\end{figure}

Fig. \ref{fig6} examines the impact of array aperture on the MI at $\text{SNR}=30\ \text{dB}$. The benchmark scheme, employing a fixed half-wavelength-spaced array, remains unchanged as its MI is aperture-independent. In contrast, the MI of both the BnB and proposed \textcolor{black}{DCSPS} schemes increases with the aperture, which expands the feasible region for antenna positioning and enhances channel reconfiguration. The proposed scheme achieves approximately 90\% of the BnB's MI at saturation. However, the growth rate of both schemes gradually declines, ultimately limited by the fixed maximum rank of the channel matrix, which is constrained by the number of antennas, despite the improved spatial characteristics provided by the movable antennas.

\begin{figure}[t]
	\centering
	\includegraphics[width=3.05in,height=2.493in]{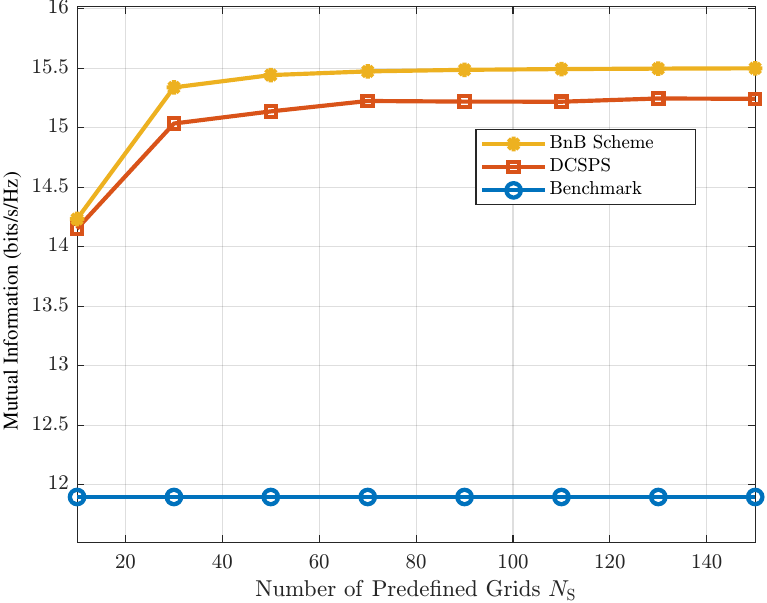}
	\caption{MI v.s. $N_{\mathrm{S}}$ under perfect CSI when $N_{\mathrm{BS}}=4$, $N_{\mathrm{U}}=2$, $\text{SNR} = 30\ \text{dB}$, and $A=2(N_{\mathrm{BS}}-1)\lambda$.}
	\label{fig7}
\end{figure}

Fig. \ref{fig7} examines the impact of the number of predefined grids on the MI under a fixed aperture, based on Fig. \ref{fig6}. While the benchmark scheme with fixed antennas remains unaffected by the grid number, both the BnB and proposed DCSPS algorithms achieve higher MI as the grid number increases, albeit with diminishing returns. A larger grid set expands the feasible antenna locations, enabling better solutions. However, for a fixed aperture, the MI upper bound remains that of the continuous MA system. Thus, increasing the number of candidate positions $N_{\mathrm{S}}$ yields progressively smaller MI gains. Jointly observing Fig. \ref{fig6} and Fig. \ref{fig7} indicates that aperture size influences the MA-system MI more significantly than the grid number.

\begin{figure}[t]
	\centering
	\includegraphics[width=3.05in,height=2.493in]{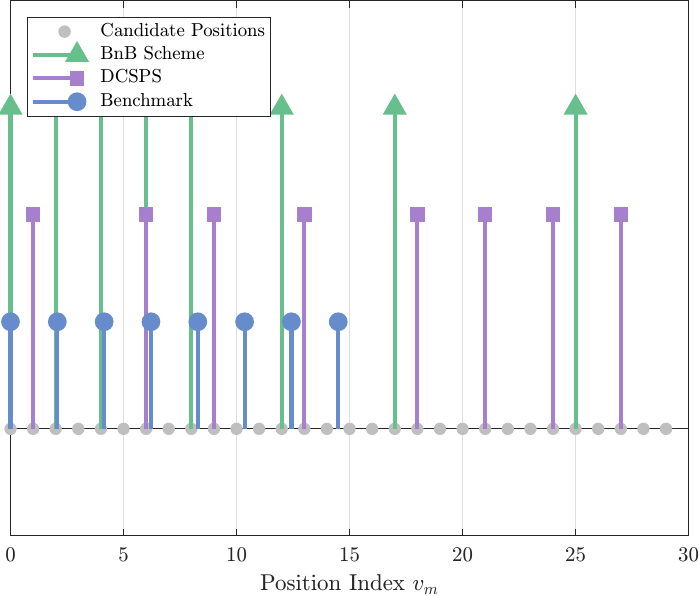}
	\caption{Antenna position design of different schemes under perfect CSI when $N_{\mathrm{BS}}=8$, $N_{\mathrm{U}}=4$, $A=(N_{\mathrm{BS}}-1)\lambda$, $N_{\mathrm{S}}=30$, and $\text{SNR} = 30\ \text{dB}$.}
	\label{fig9}
\end{figure}

Fig.~\ref{fig9} illustrates antenna positions selected by different schemes under a given channel realization. Both the BnB and proposed DCSPS schemes distribute elements across the entire aperture. Driven by its greedy selection of positions with the highest incremental gain, the DCSPS more actively explores underutilized regions, resulting in a more uniform spatial distribution than the BnB scheme. Since MA positioning enhances channel quality by projecting signal energy onto uncorrelated scattering dimensions, a larger effective aperture accesses richer spatial dimensions, offering greater channel enhancement potential.

\begin{figure}[t]
	\centering
	\makebox[\columnwidth][c]{\hspace*{-0.185in}\includegraphics[width=3.3in,height=2.493in]{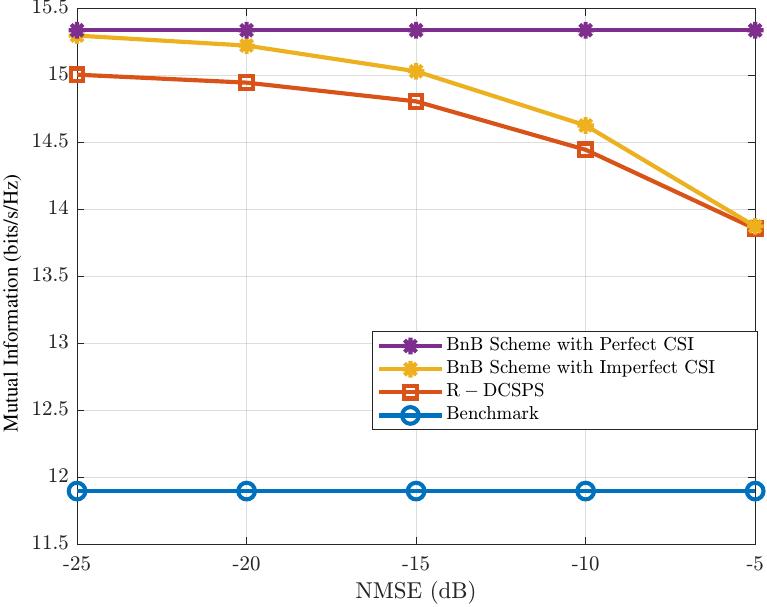}\hspace*{0.185in}}
	\caption{MI v.s. NMSE under imperfect CSI when $N_{\mathrm{BS}}=4$, $N_{\mathrm{U}}=2$, $A=2(N_{\mathrm{BS}}-1)\lambda$, $N_{\mathrm{S}}=30$, and $\text{SNR} = 30\ \text{dB}$.}
	\label{fig10}
\end{figure}


Fig. \ref{fig10} depicts the MI degradation v.s. NMSE. Bounded by the perfect-CSI BnB (upper) and fixed antenna (lower) benchmarks, the system MI declines as NMSE rises due to the increased resource cost of robustness. Despite this, the proposed scheme achieves significantly higher MI than the fixed benchmark, even under severe error conditions ($\text{NMSE}=-5\text{ dB}$). Furthermore, the proposed \textcolor{black}{R-DCSPS} approaches the MI of the imperfect-CSI BnB scheme as channel uncertainty grows, eventually matching it at $\text{NMSE}=-5\text{ dB}$, which underscores its robustness.

\subsection{Computational Complexity}
\begin{table}[t]
	\centering
	\caption{Complexity Comparisons of Different Algorithms.}
	\label{tab:complexity}
	\begin{threeparttable}
		\begin{tabular}{c|c|c|c}
			\hline
			\hline
			\makecell{ \\}  &\makecell{\bf{\textcolor{black}{DCSPS}}} &\makecell{\bf{BnB} } &\makecell{\bf{Brute-force}}  \\ \hline
			\bf{\bf{Runtime (ms)} }	           &  $5$       &$172$	& $297$          \\ \hline                           
			
			\hline
		\end{tabular}
	\end{threeparttable}
	\label{table1}
\end{table}

Table \ref{table1} compares the computational complexity of the proposed \textcolor{black}{DCSPS} algorithm, the brute-force optimal search, and a BnB method implemented in YALMIP with BMIBNB \cite{1393890}, with perfect CSI assumption. The results demonstrate that the \textcolor{black}{DCSPS} algorithm achieves significantly faster execution ($5\text{ ms}$) than both {\blue the }BnB scheme ($172\text{ ms}$, $34.4\times$ slower) and {\blue the }brute-force scheme ($297\text{ ms}$, $59.4\times$ slower). Hence, for the MI maximization MA design problem, the \textcolor{black}{DCSPS} approach can deliver a near-optimal solution at very low complexity, giving it a clear edge over alternative algorithms and making it highly attractive for practical deployment.

\begin{figure}[t]
	\centering
	\includegraphics[width=3.2in,height=2.493in]{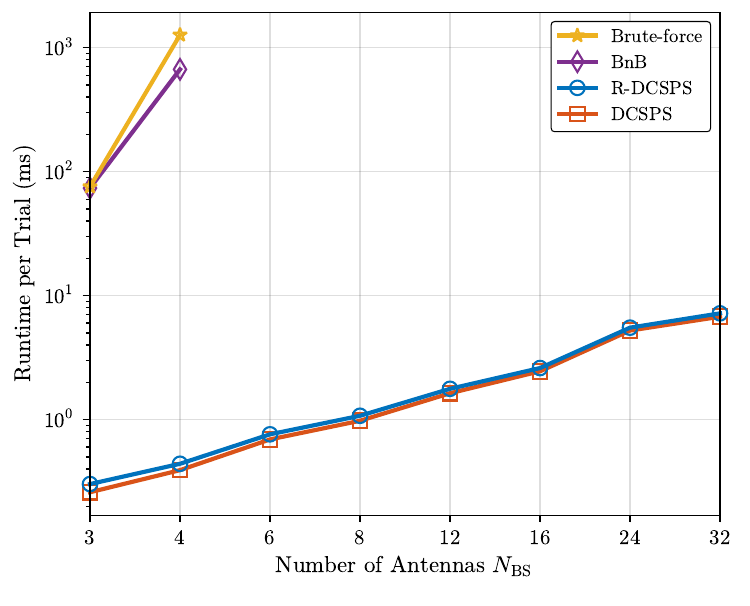}
	\caption{{\blue Runtime v.s. number of antennas $N_{\mathrm{BS}}$ when $N_{\mathrm{S}}=80$ and $K=2$, with $N_{\mathrm{U}}=2$, $L=20$, and a normalized SNR of $30$~dB.}}
	\label{fig:runtime_nbs}
\end{figure}

{\blue To evaluate the scalability of the proposed schemes, we compare the runtime of the DCSPS and R-DCSPS with the BnB and brute-force schemes as $N_{\mathrm{BS}}$ and $K$ increase. Fig.~\ref{fig:runtime_nbs} shows that the proposed DCSPS and R-DCSPS remain at the millisecond level up to $N_{\mathrm{BS}}=32$, whereas the BnB and brute-force schemes stay affordable only up to $N_{\mathrm{BS}}=4$ and become infeasible beyond it. The runtime of the proposed schemes increases with $N_{\mathrm{BS}}$, consistent with their complexity order $\mathcal{O}(N_{\mathrm{S}} N_{\mathrm{BS}}^{4})$, which is fourth order in $N_{\mathrm{BS}}$. A larger $N_{\mathrm{BS}}$ raises both the number of greedy steps and the size of matrix operations per evaluation. Even so, both proposed schemes retain a decisive complexity advantage and remain fully usable for large arrays.}

\begin{figure}[t]
	\centering
	\includegraphics[width=3.3in,height=2.493in]{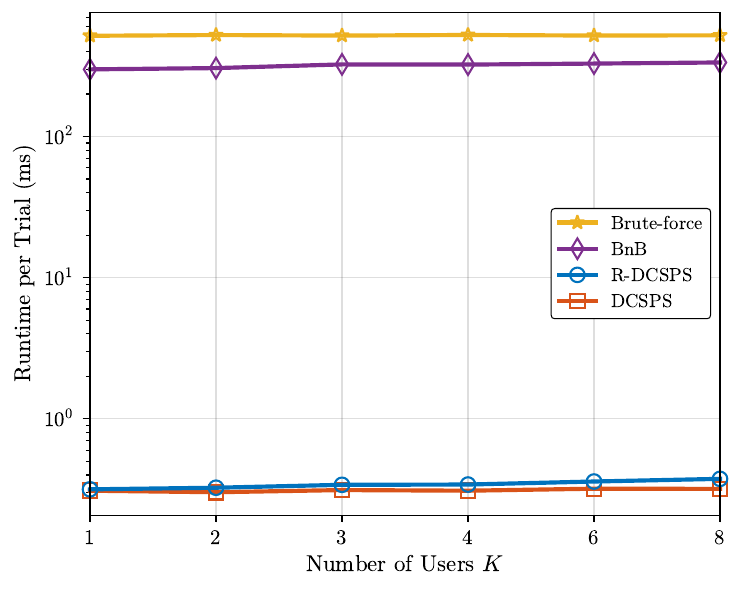}
	\caption{{\blue Runtime v.s. number of users $K$ when $N_{\mathrm{S}}=64$ and $N_{\mathrm{BS}}=4$, with $N_{\mathrm{U}}=2$, $L=20$, and a normalized SNR of $30$~dB.}}
	\label{fig:runtime_k}
\end{figure}

{\blue Fig.~\ref{fig:runtime_k} shows that the runtime stays almost flat and remains more than three orders of magnitude below the BnB and brute-force schemes, because $K$ does not change the problem dimension, i.e., the number of outer iterations, as the selection still places $N_{\mathrm{BS}}$ antennas over $N_{\mathrm{S}}$ grids. A larger $K$ only raises the cost of each objective evaluation, since the MI objective is a log-determinant whose matrix operations scale with $K N_{\mathrm{U}}$, so evaluating each candidate position becomes slightly more costly. {\blue The R-DCSPS evaluates the wider virtual channel in \eqref{34}, which appends $N_{\mathrm{BS}}$ columns representing the channel estimation error, so it stays slightly slower than the DCSPS over the whole range. Consequently, the DCSPS remains near $0.31\text{ ms}$, whereas the R-DCSPS increases only mildly from about $0.32\text{ ms}$ to about $0.38\text{ ms}$.} Overall, these results confirm that the proposed schemes retain a low polynomial complexity as $N_{\mathrm{BS}}$ and $K$ increase and stay practical for large-scale systems, whereas the BnB and brute-force schemes quickly become intractable.}

\section{Conclusion}
This paper established a novel submodular optimization framework for MA design in uplink MU-MIMO systems under different CSI assumptions. The proposed \textcolor{black}{DCSPS} algorithm achieved near-optimal MI with low complexity. Simulation results validated the effectiveness and practicality of the proposed scheme and led to several noteworthy design insights:
\begin{itemize}
	\item The MI maximization MA design problem has a monotone submodular objective, ensuring the \textcolor{black}{DCSPS} scheme optimally balances MI and complexity.
	\item Array aperture impacts MI more significantly than the number of predefined grids. Therefore, expanding aperture takes priority under limited hardware constraints.
	\item Estimation errors can be viewed as introducing virtual users that sequester communication resources originally intended for real users, thus depressing MI.
\end{itemize}

The proposed scheme provided an efficient MA design solution. {\blue Future work will first extend the proposed framework from the one-dimensional linear array to a two-dimensional planar array, over which the MAs move within a rectangular region. In this case, the MI objective remains monotone submodular, whereas the minimum spacing has to be imposed on the Euclidean distance between any two selected positions, which calls for a new characterization of the feasible position sets. Second, since the MI saturates as the array aperture increases, we will analytically characterize the MI as a function of the aperture, so that the smallest aperture attaining a prescribed fraction of the saturation MI can be determined.}

\bibliographystyle{IEEEtran}
\bibliography{IEEEabrv,IEEEreference}

\end{document}